\documentclass[12pt]{article}
\usepackage{amsmath}
\usepackage{graphicx}
\usepackage{natbib}
\usepackage{url} 

\usepackage{amssymb}
\usepackage{amsmath}
\allowdisplaybreaks
\usepackage{booktabs}
\usepackage{xcolor}
\usepackage{float}
\usepackage{subfig}
\usepackage{graphics}
\usepackage{graphicx}
\usepackage{enumitem}
\usepackage{cancel}
 \usepackage{amsthm}
\usepackage{listings}
\newtheorem{theorem}{Theorem}

\theoremstyle{definition}
\newtheorem{definition}{Definition}
\usepackage{algorithm}
\usepackage{algpseudocode}
\usepackage{ulem}
\usepackage{hyperref}

\newcommand{\blind}{0}

\date{}

\begin{document}

\def\spacingset#1{\renewcommand{\baselinestretch}%
{#1}\small\normalsize} \spacingset{1}


\if0\blind
{
  \title{\bf Local inhomogeneous 
  weighted 
  summary statistics for marked point processes 
  
  }
  \author{Nicoletta D'Angelo, Giada Adelfio\hspace{.2cm}\\
    Department of Businnes, Economics and Statistics, \\ University of Palermo, Palermo, Italy\\
    and \\
    Jorge Mateu \\
    Department of Mathematics,\\ University Jaume I, Castellon, Spain\\
    and \\
    Ottmar Cronie\thanks{
   Corresponding author; email: ottmar@chalmers.se} \\
    Department of Mathematical Sciences,\\
Chalmers University of Technology and University of Gothenburg,\\ Gothenburg, Sweden}
  \maketitle
} \fi

\if1\blind
{
  \bigskip
  \bigskip
  \bigskip
  \begin{center}
    {\LARGE\bf Local t-weighted higher order summary statistics for functional marked point processes}
\end{center}
  \medskip
} \fi

\bigskip
\begin{abstract}
We introduce a family of local inhomogeneous mark-weighted summary statistics, of order two and higher, for general marked point processes. Depending on how the involved weight function is specified, these summary statistics capture different kinds of local dependence structures. We first derive some basic properties and show how these new statistical tools can be used to construct most existing summary statistics for (marked) point processes. We then propose a local test of random labelling. This procedure allows us to identify points, and consequently regions, where the random labelling assumption does not hold, e.g.~when the (functional) marks are spatially dependent. 
Through a simulation study we show that the test is able to detect local deviations from random labelling. 
We also provide an application to an earthquake point pattern with functional marks given by seismic waveforms. 

\end{abstract}

\noindent%
\vfill

\newpage
\spacingset{1.5} 


\section{Introduction}
\label{sec:intro}

The analysis of a point pattern, given as a collection of points in a region, typically begins with computing an estimate of some summary statistic which which may be used to find specific structures in the data and suggest suitable models \citep{chiu:stoyan:kendall:mecke:13,daley:vere-jones:08,gelfand:diggle:guttorp:fuentes:10,illian:penttinen:stoyan:stoyan:08,van2000markov}.
The choice of summary statistic depends both on the pattern at hand and on the feature or hypothesis of interest.

A widely used summary statistic for descriptive analyses and diagnostics, which is obtained as an instance of the so-called reduced second moment measure \citep{cressie2001analysis,chiu:stoyan:kendall:mecke:13,moller:03}, is Ripley's $K$-function \citep{ripley:76}, which is based on the assumption of a non-marked stationary and isotropic point process. 
In the marked case, assuming discrete marks 
and 
stationarity, cross versions of the $K$- or nearest neighbour distance distribution functions 
have been proposed 
\citep{diggle:13}. For real-valued marks, the mark correlation type-functions in \cite{penttinen1989statistical,illian:penttinen:stoyan:stoyan:08} are widely used and such second order statistics have been studied in more detail and reformulated by \cite{schlather2001second}, in order to obtain a more rigorous formulation. 
However, although the assumption of stationarity is mathematically appealing,  unfortunately it can rarely be justified in practice, where, mostly, the intensity tends to change over the study region. 
This is to say that the underlying point process is inhomogeneous and, in the unmarked case, \cite{baddeley2000non} proposed an inhomogeneous extension of the $K$-function for a class of point processes, which are referred to as second order intensity-reweighted stationary. Their ideas were extended to spatio-temporal point processes in \cite{gabriel:diggle:09,moller:mohammad:12}.
Further, \cite{moller:waagepetersen:04} proposed an extension of this $K$-function to second order intensity-reweighted stationary multivariate point processes. 
As indicated in \cite{cronie2016summary,iftimi2019second}, this structure may be extended to $K$-functions for general marked point processes. 
To analyse higher order interactions in general stationary marked point processes, \cite{van2006j} proposed marked versions of the nearest neighbour distance distribution functions, the empty space function and the $J$-function. These summary statistics, which allow us to study spatial interactions between different mark groupings of the points, 
were later extended to the inhomogeneous setting by \cite{cronie2016summary,iftimi2019second}. In particular, to test for random labelling, \cite{cronie2016summary} proposed inhomogeneous Lotwick-Silverman-type Monte Carlo tests based on their new summary statistics, while \cite{iftimi2019second} proposed second order Monte Carlo tests based on permuting the attached marks. 
Further details on the random shift-type testing considered in Lotwick-Silverman-type tests can be found in \cite{mrkvivcka2021revisiting}. 

Despite the relatively long history of point process theory \citep[see e.g.][]{diggle:13,stoyan:stoyan:94,daley:vere-jones:08}, few approaches have been 
proposed 
to analyse spatial point patterns where the features of interest are functions/curves instead of qualitative or quantitative variables. 
Examples of point patterns with associated functional data include forest patterns where for each tree we have a growth function, curves representing the incidence of an epidemic over a period of time, and the evolution of distinct economic parameters such as unemployment and price rates, all for distinct spatial locations. The study of such configurations allows analysing the effects of the spatial structure on individual functions.
\cite{illian2006principal} consider for each point a transformed \cite{ripley:76}’s $K$-function to characterise spatial point patterns of ecological plant communities, whilst \cite{mateu:lorenzo:porcu:07} build new marked point processes formed by spatial locations and curves defined in terms of Local Indicators of Spatial Association (LISA) functions, which describe local characteristics of the points. 
They use this approach to classify and discriminate between points belonging to a clutter and those belonging to a feature.
Finally, the idea of analysing point patterns with attached functions has been presented coherently by \cite{comas2011second,ghorbani2021functional}.

\citet{ghorbani2021functional} introduced a very broad framework for the analysis of Functional Marked Point Processes (FMPPs), indicating how they connect the point process framework with both Functional Data Analysis (FDA; \cite{ramsay2002applied}) and geostatistics. 
In particular, they defined a new family of summary statistics, so-called \textit{weighted $n$-th order marked inhomogeneous $K$-functions}, together with their non-parametric estimators, 
which they exploited to analyse Spanish population structures, such as demographic evolution and sex ratio over time. 
This summary statistic family can be used to run a Monte Carlo test of random labelling, e.g.~by means of global envelopes test (GET; \cite{myllymaki2017global}), to assess whether the functional marks of the analysed pattern are spatially dependent.
However, this procedure is essentially global, since it does not provide information on the points which mostly contributed to the rejection of the random labelling hypothesis. 
Therefore, motivated by the need of 
detecting such points, and thus the regions in which they are located, where the functional marks really do depend on the surrounding structure, in this paper we introduce a new class of summary statistics,  \textit{local $t$-weighted marked $n$-th order inhomogeneous $K$-functions}. These are used to propose a \textit{local test of random labelling}. Here $t$ refers to a function which governs how much weight we put on different aspects of the marked point process/pattern. 

Further, we use the developed tools to analyse seismic data. 
Note that while the spatial (and temporal) locations of the epicenters of 
earthquakes are typically analysed within the framework of point processes, the associated seismic waveforms are commonly investigated in separate analyses through FDA. 
Applying the local test allows us to identify where one would expect waveforms (i.e. functional marks) to be similar to those of nearby points.

All the performed analyses 
are carried out through the \cite{R} software, and the codes are available from the first author. Preliminary data manipulation is performed through the software Python \citep{van1995python}.
 
The structure of the paper is as follows.
In Section \ref{sec:data}, the motivation of this paper is presented, showing the dataset and problem that will be further analysed along the paper.
Section \ref{sec:preliminaries} contains some preliminaries on functional marked point processes. In Section \ref{sec:proposal}, we present our proposed local t-weighted 
$n$-th order 
inhomogeneous $K$-functions and their main properties, also relating them to their global counterparts. Section \ref{sec:test} outlines the main steps to run a local test of random labelling. 
In Section \ref{sec:sims}, we present a motivating example to show the further advantages of a local test, compared to a 
global 
one. 
To have a comprehensive understanding of the performance of the proposed local test, we show simulation results under different scenarios.
Section \ref{sec:real} provides an application to seismic data.
Finally, conclusions are drawn in Section \ref{sec:conc}.

\section{Data and motivation }
\label{sec:data}


\begin{figure}[htb]
	\centering
		\includegraphics[width=.75\textwidth]{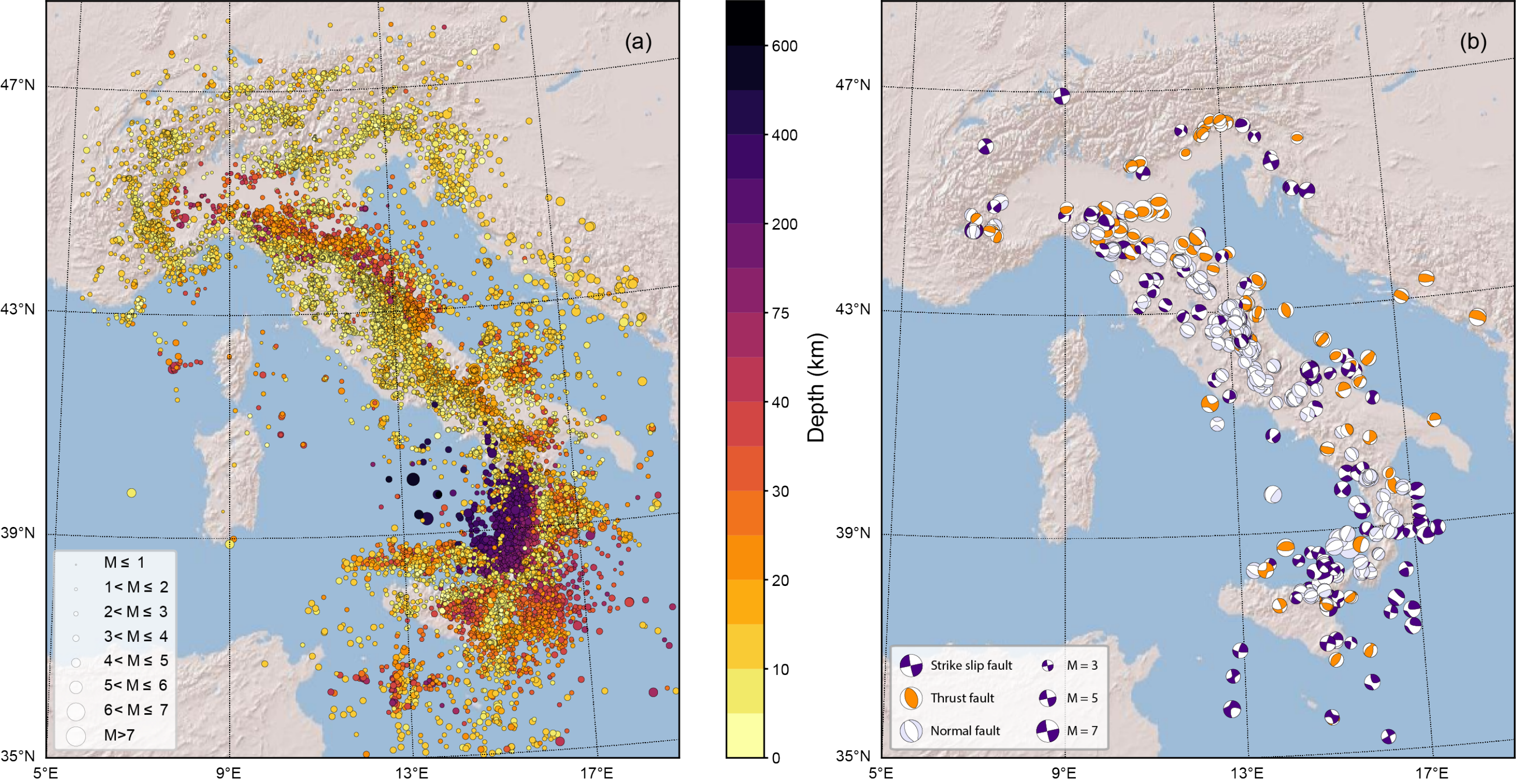}
\includegraphics[width=.75\textwidth]{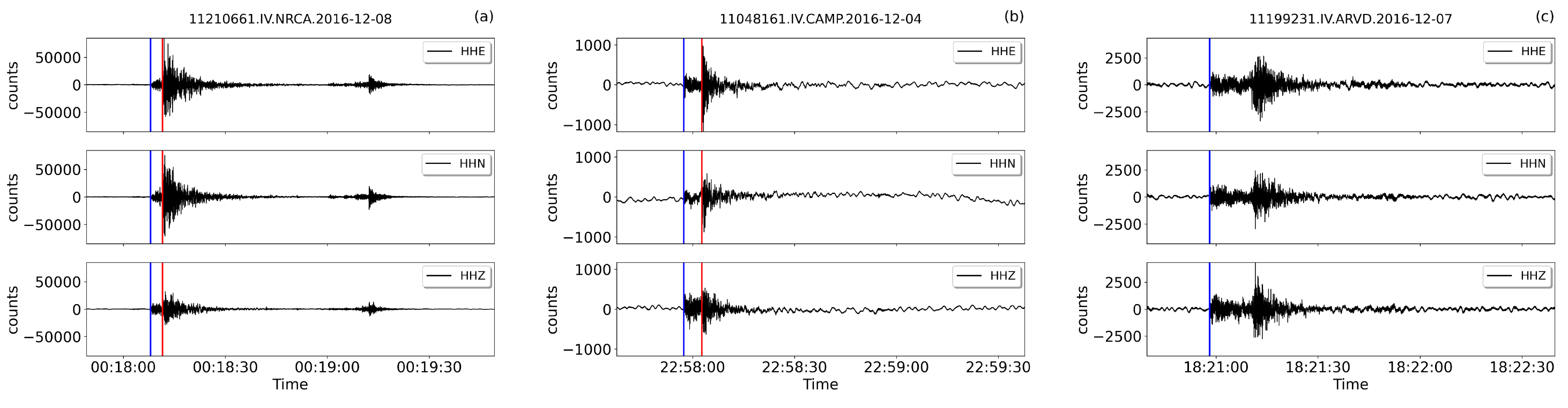}
	\caption{\textit{Top-left panel:} Earthquake locations; \textit{Top-right panel:} Seismic stations used for waveforms extraction. The symbol sizes are proportional to earthquake magnitude and number of arrival phases recorded by stations, respectively; \textit{Bottom panels}: Seismic waveforms of some events with magnitude in the range $[2, 4]$.
	 \textit{Source: \cite{michelini2021instance}}.}
	\label{fig:instance}
\end{figure}

Earthquakes' detection
provides a whole set of data which are usually studied separately, i.e. spatial (and temporal) occurrence 
of points through point process theory (\cite{siino2017spatial,iftimi2019second,dangelo2021locall}, to cite just a few recent works), 
and the analysis of waveforms through FDA \citep{adelfio2011fpca,adelfio2012simultaneous,chiodi2013clustering}.

A recently released set of data on Italian seismic activity encompasses both of these data types. 
\textit{The Italian seismic dataset for machine learning (INSTANCE)} is a dataset of seismic waveforms data and associated metadata \citep{michelini2021instance}, which includes 54008 earthquakes for a total of 1159249 3-channel waveforms. It also contains 132330 3-channel noise waveforms.
For each of these waveforms, 
115 metadata (i.e.\ statistical variables) are available, providing information on station, trace, source, path and quality. Overall, the data are collected by 19 networks which consist of 620 seismic stations. 
The dataset is available on \url{http://www.pi.ingv.it/instance/}.

The earthquake list 
in the dataset
is based on the Italian seismic bulletin (\url{http://terremoti.ingv.it/bsi}) of the ``Istituto Nazionale di Geofisica e Vulcanologia'', 
includes events which occurred 
between January 2005 and January 2020, and in the magnitude range between 0.0 and 6.5. The waveform data have been recorded primarily by the Italian National Seismic Network. The top panels of Figure \ref{fig:instance} depict the earthquake locations and the seismic stations which recorded the events.

In the bottom panels of Figure \ref{fig:instance}, some waveforms contained in the dataset are represented.
All the waveform traces have a length of 120 seconds, are sampled at 100 Hz, and are provided both in counts and ground motion physical units after deconvolution of the instrument transfer functions. The waveform dataset is accompanied by metadata consisting of more than 100 variables providing comprehensive information on the earthquake source, the recording stations, the trace features, and other derived quantities.


\section{Preliminaries on marked point processes}
\label{sec:preliminaries}

Throughout the paper, we consider marked point processes $Y=\{(x_i,m_i)\}_{i=1}^N$ 
\citep[Definition 6.4.1]{daley:vere-jones:08}, with ground points $x_i$ in the $d$-dimensional Euclidean space $\mathbb{R}^d$, $d\geq1$, which is equipped with 
the Lebesgue measure $\vert A \vert =\int_A\text{d}z$ for Borel sets $A\in\mathcal{B}(\mathbb{R}^d)$; a closed Euclidean $r$-ball around $x\in\mathbb{R}^d$ will be denoted by $b[x,r]$. 
By definition, the ground process $Y_g=\{x_i\}_{i=1}^N$, obtained from $Y$ by ignoring the marks, is a well-defined point process on $\mathbb{R}^d$ in its own right. We shall assume that $Y$ is simple, that is, it almost surely (a.s.) does not contain multiple points. Note that, formally, $Y$ is a random element in the measurable space $(N_{lf},\mathcal{N})$ of locally finite point configurations/patterns $\mathbf{x}=\{((x_1,m_1),\ldots, (x_n,m_n))\}$, $n\geq0$ \citep{daley:vere-jones:08,van2000markov}.
We assume that the mark space $\mathcal{M}$ is Polish and equipped with a finite reference measure $\nu$ on the Borel $\sigma$-algebra $\mathcal{B(M)}$. The Borel
$\sigma$-algebra  $\mathcal{B}(\mathbb{R}^d \times  \mathcal{M})=\mathcal{B}(\mathbb{R}^d)\otimes\mathcal{B}(\mathcal{M})$ is endowed with the product measure $A\times E\mapsto \vert A \vert \nu(E)$, $A\times E\in\mathcal{B}(\mathbb{R}^d \times  \mathcal{M})$. 
We will let $Y(A\times E)=\sum_{(x,m)\in Y}\mathbf{1}\{(x,m)\in A\times E\}$, where $\mathbf{1}$ is the indicator function, denote the cardinality of the random set $Y\cap(A\times E)$. 

Given this general setup, one may obtain various forms of marked point processes, most notably multivariate/multitype point processes with  $\mathcal{M}=\{1,\ldots,k\}$ \citep{diggle:13} and functional marked point processes with $\mathcal{M}$ given by a suitable function space \citep{ghorbani2021functional}.

\subsection{Functional Marked Point Processes}

In this section, we provide the definition of functional marked point processes following \cite{ghorbani2021functional}.


In classical FDA, one analyses a collection of functions $\{f_1(t), \ldots, f_n(t)\}$, $t \in \mathcal{T} \subset [0,\infty)$, $n \geq 1$, which take values in some Euclidean space $\mathbb{R}^k$, $k \geq 1$, and belong to some suitable function space, typically an $L_2$-space. Although $t$ usually represents time, it could also represent some other quantity, for example, spatial distance. 
Classically, one would assume that such a collection of functions constitute realisations or samples of some collection of independent and identically distributed (iid) random functions or stochastic processes $\{ F_1(t), \ldots , F_n(t)\}$, $t \in \mathcal{T}$. 
Such an assumption may, however, be questioned in certain settings. 
For example, two functions $f_i$ and $f_j$, which are spatially close to each other in $\mathbb{R}^k$, could  gain (or lose) from being close to each other. Accordingly, it seems natural to relax the iid
assumption for $F_1,\ldots, F_N$. A natural way to handle such a scenario is to generate $F_1,\ldots, F_N$ conditionally on some collection of (dependent) random spatial locations. 
Note that the conditional distribution of $F_1,\ldots, F_N$ could render them both dependent and independent. 

To facilitate such a setting, we consider a functional marked point process \citep{ghorbani2021functional}, which is defined as a marked point process where the marks are random elements in some (Polish) function space, $\mathcal{M}$, most notably the space of $L_2$-functions $f:\mathcal{T}\to\mathbb{R}^k$. Realisations of FMPPs are called functional
marked point patterns. 
It is noteworthy that the original formal construction of functional marked point processes by \citet{ghorbani2021functional} also included an additional non-functional mark, so that each ground process point would be marked by a pair which consists of a function and a non-functional variable. We here do not consider such auxiliary non-functional marks.



\subsection{Product densities}

Provided that it exists, the \textit{$n$-th order intensity/product density} function $\rho^{(n)}$, $n\geq1$, which is the density of the $n$-th order factorial moment measure $\alpha^{(n)}$, may be specified through the $n$-th order \textit{Campbell formula}. It states that, for any non-negative measurable function $h$ on $(\mathbb{R}^d \times  \mathcal{M})^n$, the expectation of the random sum of $h$ over $n$-tuples of distinct points of $Y$ satisfies
\begin{equation}
\begin{split}
  & \mathbb{E}\left[\sum^{\ne}_{(x_1,m_1), \ldots, (x_n,m_n) \in Y} h((x_1,m_1), \ldots, (x_n,m_n)) \right] = \\
  & =  \int \cdots \int  h((x_1,m_1), \ldots, (x_n,m_n))  \rho^{(n)}((x_1,m_1), \ldots, (x_n,m_n)) \prod_{i = 1}^{n} \text{d}x_i \nu(dm_i).
    \label{eq:campmark}
\end{split}
\end{equation}
Heuristically, $\rho^{(n)}((x_1,m_1), \ldots, (x_n,m_n)) \text{d}x_1 \nu(dm_1)\cdots\text{d}x_n \nu (dm_n)$ gives the probability that $Y$ has points in infinitesimal neighbourhoods $d(x_i,m_i)\ni (x_i,m_i) \in \mathbb{R}^d \times \mathcal{M}$ with measures $\text{d}x_i \nu(dm_i)$, $i=1,\ldots,n$. 
Moreover, we retrieve $\alpha^{(n)}((A_1\times E_1)\times\cdots\times(A_n\times E_n))$, $(A_i\times E_i)\in\mathcal{B}(\mathbb{R}^d \times  \mathcal{M})$, $i=1,\ldots,n$, by letting $h$ be given by the indicator function  $\mathbf{1}\{(x_1,m_1)\in (A_1\times E_1),\ldots,(x_n,m_n)\in (A_n\times E_n)\}$.
It further follows that 
$$
\rho^{(n)}((x_1,m_1), \ldots, (x_n,m_n))
=
f_{x_1,\dots,x_n}(m_1, \ldots, m_n)
\rho_g^{(n)}(x_1,\ldots,x_n),
$$
where $\rho_g^{(n)}$ is the $n$-th order product density of $Y_g$ and $f_{x_1,\dots,x_n}(\cdot)$ is a conditional density function on $\mathcal{M}^n$ which governs the joint distribution of $n$ marks, given that their associated ground process points are given by $x_1,\dots,x_n\in\mathbb{R}^d$. These, in turn, yield the corresponding mark distributions
\begin{equation*}
 M^{x_1,\ldots,x_n}(E_1, \ldots, E_n) = \int_{E_1} \cdots  \int_{E_n} f_{x_1,\dots,x_n}(m_1, \ldots, m_n) \prod_{i = 1}^{n} \nu(dm_i),
\end{equation*}
which govern the joint distribution on $n$ marks, given the associated ground process locations.

The intensity measure of $Y$, which coincides with the first order factorial moment measure, here satisfies 
\begin{align}
\alpha(A \times E) =& \mathbb{E}[Y(A \times E)] 
=
\int_A \int_E \rho(x,m) \text{d}x \nu (dm)
\nonumber
\\
=& 
\int_A \int_E  f_{x}(m) \rho_g(x) \text{d}x \nu (dm)
=
\int_A M^{x}(E)\rho_g(x) \text{d}x
,
\label{eq:first}
\end{align}
where the first order intensity functions $\rho=\rho^{(1)}$ and $\rho_g=\rho_g^{(1)}$ are typically referred to as \textit{the intensity functions} of $Y$ and $Y_g$. 
Note that 
$\rho$ may be viewed as a ``heat map" which reflects the infinitesimal chance of having a point of $Y$ at/around an arbitrary location in $\mathbb{R}^d \times  \mathcal{M}$. When the intensity function (of the ground process) is constant, we say that the (ground) process is homogeneous, otherwise it is called inhomogeneous.  

When, conditional on the ground process, all marks have the same marginal univariate distribution, so that $M^{z}(E) = \int_{E} f_{z}(m) \text{d}\nu(dm) = \int_{E} f(m) \text{d}\nu(m) = M(E)$, we say that $X$ has a common (marginal) mark distribution. 
This holds e.g.~when $Y$ is stationary, i.e.~when its distribution is invariant under translations of the ground points; here $\alpha(A \times E)=\rho_g M(E) \vert  A \vert $ and $\rho_g>0$ is the constant intensity of the ground process. 
We will see that, at times, it is particularly convenient to have here that the reference measure $\nu$ coincides with the common mark distribution $M$, which implies that the common mark density $f$ is set to $1$ and $\rho(x,m)=\rho_g(x)$. 

When $Y$ is independently marked, i.e.~when the marks are independent conditional on the ground process, 
$f_{x_1,\ldots,x_n}(m_1, \ldots, m_n)=f_{x_1}(m_1)\cdots f_{x_n}(m_n)$ for any $n\geq1$ and if, in addition, there is a common mark distribution, whereby the marks are iid conditional on the ground process, we say that $Y$ is randomly labelled and note that $f_{x_1,\ldots,x_n}(m_1, \ldots, m_n)=f(m_1)\cdots f(m_n)$.

\subsubsection{Intensity reweighted stationarity}

We next turn to the notion of a $k$-th order marked intensity reweighted stationary ($k$-MIRS) marked point process $Y$ \citep{ghorbani2021functional}. We say that $Y$ is $k$-MIRS, $k\in\{1,2,\ldots\}$, if $\rho$ is bounded away from 0 and the correlation functions
\begin{align*}
&g^{(n)}((x_1,m_1), \ldots, (x_n,m_n))
=
\\
=&
\frac{\rho^{(n)}((x_1,m_1), \ldots, (x_n,m_n))}{\rho(x_1,m_1)\cdots\rho(x_n,m_n)}
=
\frac{f_{x_1,\dots,x_n}(m_1, \ldots, m_n)}{f_{x_1}(m_1)\cdots f_{x_n}(m_n)}
\frac{\rho_g^{(n)}(x_1, \ldots, x_n)}{\rho_g(x_1)\cdots\rho_g(x_n)}
,
\quad n\geq1,
\end{align*}
satisfy $g^{(n)}((x_1,m_1), \ldots, (x_n,m_n))=g^{(n)}((z+x_1,m_1), \ldots, (z+x_n,m_n))$ for any $z\in\mathbb{R}^d$ and any $n\leq k$. Note that $g^{(1)}(\cdot)\equiv1$ and that the second ratio on the right hand side is the $n$-th order correlation function, $g_g^{(n)}$, of the ground process. Provided that the product densities of all orders exist, stationarity implies $k$-MIRS for all orders $k\geq1$. Note further that $g^{(n)}(\cdot)\equiv1$, $n\geq1$, for a Poisson process and when $g^{(n)}((x_1,m_1), \ldots, (x_n,m_n))>1$ points of $Y_g$ in infinitesimal neighbourhoods of $x_1,\ldots,x_n$ with marks in infinitesimal neighbourhoods of  $m_1,\ldots,m_n$ tend to cluster/aggregate. Similarly, $g^{(n)}((x_1,m_1), \ldots, (x_n,m_n))<1$ indicates inhibition/regularity.

\subsection{Palm distributions}

Let $Y$ be a simple marked point process whose intensity function exists. Many of the summary statistics we will consider can be expressed in terms of \textit{reduced Palm distributions}. These satisfy the \textit{reduced Campbell–Mecke} formula which states that, for any non-negative measurable function $h$ on the product space $(\mathbb{R}^d \times \mathcal{M})\times N_{lf}$,  
\begin{align}
\label{eq:mecke}
\mathbb{E} \left[ \sum_{(z,m) \in Y} h((z,m), Y \backslash \{ (z,m) \}) \right]
=& 
\int
\mathbb{E}[h((x,m),Y^{!(x,m)})]\rho(x,m)\text{d}x
\nu(dm)
\\
=& 
\int
\mathbb{E}^{!(x,m)}[h((x,m),Y)]\rho(x,m)\text{d}x
\nu(dm).
\nonumber
\end{align}
Here $Y^{!(x,m)}$ is the reduced Palm process at $(x,m)\in\mathbb{R}^d\times\mathcal{M}$, which we interpret as $Y$ conditioned on the null event that there is a point in $(x,m)$, which is removed upon realisation. The probability distribution $P^{!(x,m)}(\cdot) = \mathbb{P}^{!(x,m)}(Y \in \cdot) = \mathbb{P}(Y^{!(x,m)}\in\cdot)$ on $(N_{lf},\mathcal{N})$, which corresponds to $\mathbb{E}^{!(x,m)}$, is called the reduced Palm distribution at $(x,m)$. 

\section{Local weighted marked summary statistics}\label{sec:proposal}

Global summary statistics have had a prominent role in the statistical analysis of point processes. More precisely, their non-parametric estimators are typically used to characterise the degree of spatial interaction present in the underlying data-generating point process. 
In Section \ref{sec:intro}, we have reviewed a few such examples, for instance $K$-functions. 


The individual contributions to a global statistic, which are commonly called Local Indicators of Spatial Association (LISA) functions, can be used to identify outlying components measuring the influence of each contribution to the global statistic \citep{anselin:95}. This is the case of the scatter plot based on the local Moran index \citep{anselin1996chapter}. On the other hand, the individual contributions
can be used to test for specific local structures, such as 
spatial association and hot spot detection in areal data  \citep{getis:ord:92}. Basically, the local statistics mentioned so far are often used to analyse areal data but \cite{getis:franklin:87} introduced a local version of the $K$-function for spatial point processes to show that trees exhibit different kinds of heterogeneity when examined at different scales of analysis. The notion of individual functions for 
certain 
statistics has also been studied in \cite{stoyan:stoyan:94} and \cite{mateu:lorenzo:porcu:10} showed that the local product density function \citep{cressie2001analysis} is more sensitive to identifying different local structures and unusual points than the local $K$-function. Applications of LISA functions range from detection of features in images with noise \citep{mateu:lorenzo:porcu:07} to detection of disease clusters \citep{moraga:montes:11}. In \cite{siino2018testing} the authors extend local indicators of spatial association to the spatio-temporal context (LISTA functions) based on the second order product density, and these local functions have been used to define a proper statistical test for clustering detection.
Recently, LISTA functions have been used both for diagnostic \citep{adelfio2020some} and fitting purposes \citep{dangelo2021locally}.
Finally, \cite{dangelo2021assessing} extended LISTA functions to spatio-temporal point processes living on linear networks.

As we have clearly indicated, 
an alternative to studying the aforementioned global summary statistics for marked point processes 
is 
considering 
local summary statistics which 
describe the spatial interaction in the vicinity of a given marked point. 
In order to do so here in the marked context, we introduce the function
\begin{align}
\label{e:GeneralLocal}
  &L((x,m),\mathbf{x}) = 
  L_n((x,m),\mathbf{x}; \tilde t, \tilde \rho) 
  =
  \sum_{(x_1,m_1), \ldots, (x_{n-1},m_{n-1}) \in \mathbf{x} 
  }^{\ne}
  \frac{
  \tilde t((x,m),(x_1,m_1),\ldots,(x_{n-1},m_{n-1}))
  }{\tilde \rho(x,m)\tilde \rho(x_1,m_1) \cdots \tilde \rho(x_{n-1},m_{n-1})}
  ,
\end{align}
for $(x,m)\in \mathbb{R}^d\times\mathcal{M}$, point patterns $\mathbf{x}\in N_{lf}$ and measurable $\tilde t:(\mathbb{R}^d\times\mathcal{M})^n\to\mathbb{R}$, 
$n\geq2$. Note that, formally, the argument $\tilde \rho$ does not need to be the true intensity function $\rho$ of $Y$, it could e.g.~be a plug-in estimator. We will exploit Definition \ref{def:LocalMarked}, and thereby \eqref{e:GeneralLocal}, to define proper notions of (mark-weighted $n$-th order inhomogeneous) local summary statistics.

\begin{definition}
\label{def:LocalMarked}
Given a marked point process $Y$, we refer to $L((x,m),Y\setminus\{(x,m)\}; \tilde t, \tilde \rho)$, $(x,m)\in Y$, as the family of $n$-th order local marked cumulative summary statistics of $Y$ associated with $\tilde t$ and $\tilde \rho$. 
\end{definition}

The construction of a specific local statistic is obtained by identifying when, for some function family $\{\tilde t_r\}$, 
\begin{align}
\label{e:GeneralGlobal}
G(r,Y) = \sum_{(x,m)\in Y} L_n((x,m), Y\setminus\{(x,m)\}; \tilde t_r, \tilde\rho)
\end{align}
forms an estimator of an existing global summary statistic. 

Using $n$-th order local marked cumulative summary statistics 
to quantify local spatial interactions for a point pattern $\mathbf{x}$ entails inserting an estimate $\widehat\rho(x,m)=\widehat f_{z}(m)\widehat{\rho}_g(x)$ for the unknown intensity $\rho(x,m)=f_{z}(m)\rho_g(x)$, i.e.\ setting $\tilde\rho=\widehat\rho$. 
When we assume that there is a common mark distribution which coincides with the mark reference measure $\nu$, we obtain that $\widehat\rho(x,m)=\widehat{\rho}_g(x)$, i.e.\ the intensity estimate does not depend on the mark values. 
Imposing this assumption is particularly convenient when dealing with functional marks since estimation of the mark density, which here is a density on a function space, is rather challenging and beyond the scope of this paper. Note that when $Y$ is randomly labelled, it has a common mark distribution and in this setting the assumption  $\widehat\rho(x,m)=\widehat{\rho}_g(x)$ thus makes sense.

Turning to the distributional properties of the $n$-th order local marked cumulative summary statistics, we next derive their expectations under the assumption of $k$-MIRS. Note, in particular, that the choice of $\tilde t$ plays a significant role here.

\begin{theorem}
\label{thm:Local}
When $Y$ is $k$-MIRS and $\tilde\rho=\rho$, for any $W\in\mathcal{B}(\mathbb{R}^d)$ the expectation of $L((x,m),Y\setminus\{(x,m)\}\cap W\times\mathcal{M}; \tilde t, \rho)$, $(x,m)\in Y\cap W\times\mathcal{M}$, is almost everywhere given by 
\begin{align*}
&\int_{W-x}\cdots\int_{W-x}
 \Bigg(
 \int_{\mathcal{M}}\cdots\int_{\mathcal{M}}
  \tilde t((x,m),(x_1+x,m_1),\ldots,(x_{n-1}+x,m_{n-1}))
  \times
  \\
  &\times
 \frac{f_{0,x_1,\ldots,x_{n-1}}(m,m_1, \ldots, m_{n-1})}{f_{0}(m) f_{x_1}(m_1)\cdots f_{x_{n-1}}(m_{n-1})}
 \nu(dm_1) \cdots \nu(dm_{n-1})
 \Bigg)
 g_g^{(n)}(0, x_1, \ldots, x_{n-1})
 \text{d}x_1 \cdots \text{d}x_{n-1}
  ,
\end{align*}
when $2\leq n\leq k$. Moreover, the expectation of $G(r,Y\cap W\times\mathcal{M})$ is obtained by replacing $\tilde t$ by $\tilde t_r$ in the expression above and integrating it over $W\times\mathcal{M}$ with respect to the reference measure on $\mathbb{R}^d \times \mathcal{M}$. 
\end{theorem}

\begin{proof}
Note first that the expectation coincides with 
\begin{align*}
&
\mathbb{E}^{!(x,m)}[L_n((x,m), Y\cap W\times\mathcal{M}; \tilde t, \rho)]=
\\
=&
\mathbb{E}^{!(x,m)}\left[
\sum_{(x_1,m_1), \ldots, (x_{n-1},m_{n-1}) \in Y\cap W\times\mathcal{M}}^{\ne} 
\frac{
  \tilde t((x,m),(x_1,m_1),\ldots,(x_{n-1},m_{n-1}))
  }{\rho(x,m) \rho(x_1,m_1) \cdots \rho(x_{n-1},m_{n-1})}
\right]
.
\end{align*} 
Hence, our starting point will be the reduced Campbell-Mecke formula. 
Consider an arbitrary bounded $A\times E\in\mathcal{B}(\mathbb{R}^d \times \mathcal{M})$. It follows that 
\begin{align*}
  & 
  \mathbb{E}\left[\sum_{(x,m)\in Y \cap A \times E}  \sum_{(x_1,m_1), \ldots, (x_{n-1},m_{n-1}) \in Y \setminus \{ (x,m) \}\cap W\times\mathcal{M}}^{\neq} 
  \frac{
  \tilde t((x,m),(x_1,m_1),\ldots,(x_{n-1},m_{n-1}))
  }{\rho(x,m) \rho(x_1,m_1) \cdots \rho(x_{n-1},m_{n-1})}
  \right]
  =
  \\
  =& 
  \int_{A\times E}
  \mathbb{E}^{!(x,m)} \left[
\sum_{(x_1,m_1), \ldots, (x_{n-1},m_{n-1}) \in Y\cap W\times\mathcal{M}}^{\ne} \frac{
  \tilde t((x,m),(x_1,m_1),\ldots,(x_{n-1},m_{n-1}))
  }{\rho(x_1,m_1) \cdots \rho(x_{n-1},m_{n-1})}\right]
  \text{d}x \nu(dm).
\end{align*}
On the other hand, by the Campbell formula we have that
\begin{align*}
  & 
  \mathbb{E}\left[\sum_{(x,m)\in Y \cap A \times E}  \sum_{(x_1,m_1), \ldots, (x_{n-1},m_{n-1}) \in Y \setminus \{ (x,m) \}\cap W\times\mathcal{M}}^{\neq} 
  \frac{
  \tilde t((x,m),(x_1,m_1),\ldots,(x_{n-1},m_{n-1}))
  }{\rho(x,m) \rho(x_1,m_1) \cdots \rho(x_{n-1},m_{n-1})}
  \right]
  =
  \\
  =& 
  \int_{A\times E}\int_{\mathbb{R}^d \times \mathcal{M}}\cdots\int_{\mathbb{R}^d \times \mathcal{M}}
  \mathbf{1}\{x_1,\ldots,x_{n-1}\in W\}
  \tilde t((x,m),(x_1,m_1),\ldots,(x_{n-1},m_{n-1}))
  \times
  \\
  &\times
  g^{(n)}((x,m),(x_1,m_1),\ldots,(x_{n-1},m_{n-1}))
  \text{d}x_1 \nu(dm_1)
  \cdots
  \text{d}x_{n-1} \nu(dm_{n-1})
  \text{d}x \nu(dm)
  \\
  =& 
  \int_{A\times E}
  \int_{\mathbb{R}^d \times \mathcal{M}}\cdots\int_{\mathbb{R}^d \times \mathcal{M}}
  \prod_{i=1}^{n-1}\mathbf{1}\{u_i\in W-x\}
  \tilde t((x,m),(u_1+x,m_1),\ldots,(u_{n-1}+x,m_{n-1}))
  \times
  \\
  &\times
  g^{(n)}((0,m),(u_1,m_1),\ldots,(u_{n-1},m_{n-1}))
  \text{d}u_1 \nu(dm_1)
  \cdots
  \text{d}u_{n-1} \nu(dm_{n-1})
  \text{d}x \nu(dm)
\end{align*}
by the imposed $k$-MIRS and a change of variables, $u_i+x=x_i$. 
Hence, since $A\times E\in\mathcal{B}(\mathbb{R}^d \times \mathcal{M})$ was arbitrary, for almost every $(x,m)$ we have that 
\begin{align*}
&
\mathbb{E}^{!(x,m)}[L_n((x,m), Y; \tilde t, \rho)]=
\\
=&
\int_{(W-x) \times \mathcal{M}}\cdots\int_{(W-x) \times \mathcal{M}}
  \tilde t((x,m),(u_1+x,m_1),\ldots,(u_{n-1}+x,m_{n-1}))
  \times
  \\
  &\times
  g^{(n)}((0,m),(u_1,m_1),\ldots,(u_{n-1},m_{n-1}))
  \text{d}u_1 \nu(dm_1)
  \cdots
  \text{d}u_{n-1} \nu(dm_{n-1})
 \\
 =&
 \int_{W-x}\cdots\int_{W-x}
 \Bigg(
 \int_{\mathcal{M}}\cdots\int_{\mathcal{M}}
  \tilde t((x,m),(u_1+x,m_1),\ldots,(u_{n-1}+x,m_{n-1}))
  \times
  \\
  &\times
 \frac{f_{0,u_1,\ldots,u_{n-1}}(m,m_1, \ldots, m_{n-1})}{f_{0}(m) f_{u_1}(m_1)\cdots f_{u_{n-1}}(m_{n-1})}
 \nu(dm_1) \cdots \nu(dm_{n-1})
 \Bigg)
 g_g^{(n)}(0, u_1, \cdots, u_{n-1})
 \text{d}u_1 \cdots \text{d}u_{n-1} 
  ,
\end{align*}
by Fubini's theorem.

\end{proof}

The first thing we note is that when $Y$ is independently marked then the density ratio in the expression for the expectation vanishes. In addition, if $Y$ is a Poisson process on $\mathbb{R}^d \times \mathcal{M}$ which satisfies being a marked point process with mark space $\mathcal{M}$, then the expectation reduces to an integral with $\tilde t$ as integrand. These observations may be used as benchmarks for when $Y$ exhibits mark (in)dependence and spatial interaction locally.


\subsection{Special cases}

We next illustrate how \eqref{e:GeneralGlobal}, through Definition \ref{def:LocalMarked} and \eqref{e:GeneralLocal}, reduces to several existing summary statistic estimators by varying $\tilde t$ and $\tilde \rho$. 


\subsubsection{Ground \textit{K}-functions}
First, set $n=2$ and $\tilde t$ to $\tilde t_r((x,m),(x_1,m_1))=w(x,x_1)\mathbf{1}\{x_1\in x+C\}/ \vert W \vert $, $r\geq0$, where $W\subseteq\mathbb{R}^d$, $\vert  W \vert  >0$, and $w(\cdot)$ is an edge correction term. 
If the ground process is stationary with intensity $\rho_g>0$ and $\tilde \rho(x,m)\equiv\rho_g$, then \eqref{e:GeneralGlobal} with $Y$ set to $Y\cap W\times\mathcal{M}$ reduces to an estimator of Ripley's $K$-function when $x+C=x+b[0,r]=b[x,r]$ whereas if the ground process is inhomogeneous and we set $\tilde \rho(x,m)=\rho_g(x)$, it follows that \eqref{e:GeneralGlobal} reduces to an estimator of the inhomogeneous $K$-function \citep{baddeley2000non} for $Y_g$.
The extension to space-time is straightforward; replace the Euclidean ball $b[0,r]$ by 
$C = \{ (x, s) : \| x \| \leq r, |s| \leq t\}\in \mathcal{B}(\mathbb{R}^{d+1})$, where $\| \cdot \|$ denotes the Euclidean norm \citep{cronie:lieshout:15,gabriel:diggle:09,iftimi2019second}. 

\subsubsection{Marked \textit{K}-functions}
When $n=2$, by instead letting $\tilde t_r((x,m),(x_1,m_1))=w(x,x_1)\mathbf{1}\{x_1\in x+C\}\mathbf{1}\{m\in E, m_1 \in E_1\}/(|W|
\nu(E)\nu(E_1))$ and $\tilde\rho=\rho$ in \eqref{e:GeneralLocal}, using a suitable edge correction function $w(\cdot)$, then $G(r,Y\cap W\times \mathcal{M})$ in \eqref{e:GeneralGlobal} 
reduces to an estimator of the \textit{marked second order reduced moment measure} $\mathcal{K}^{E E_1}(C)$ of \cite{iftimi2019second}, which measures the intensity reweighted interactions between points with marks in $E$ and points with marks in $E_1$, when their separation vectors belong to $C \in \mathcal{B}(\mathbb{R}^d)$. We note that measures of this kind are in general not symmetric, i.e.\  $\mathcal{K}^{ E E_1}(\cdot) \ne \mathcal{K}^{E_1 E}(\cdot)$ \citep{iftimi2019second}. Furthermore, choosing $C$ to be the closed origin-centred ball $b[0,r]$ of radius $r \geq 0$, we consider the marked inhomogeneous $K$-function $K^{E  E_1}_{inhom}(r)$ of \cite{cronie2016summary}, which measures pairwise intensity reweighted spatial dependence within distance $r$ between points with marks in $E$ and points with marks in $E_1$. 
    
By additionally letting $n>2$, we obtain a definition of a \textit{marked $n$-th order reduced moment measure}, $\mathcal{K}^{ E \times_{i=1}^{n-1}E_i}(C_1 \times \cdots \times C_{n-1})$, which measures the intensity reweighted spatial interaction between an arbitrary point with mark in $E$ and distinct $(n-1)$-tuples of other points, where the separation vectors between the $E$-marked point and these $n-1$ points, which have marks in $E_1,\ldots,E_{n-1}$, belong to $C_1,\ldots,C_{n-1}$. 
We note that $C_i = b[0,r]$, $i = 1, \ldots, n-1$, $r \geq 0$, yields an $n$-point version of the marked inhomogeneous $K$-function $K_{inhom}^{ E \times_{i = 1}^{n-1} E_i}(r)$ of \cite{cronie2016summary}, which may be used to analyse intensity reweighted interactions between a point with mark in $ E$ and $n-1$ of its $r$-close neighbours, which have marks belonging to the respective sets $E_1,\ldots,E_{n-1}$. 

\subsubsection{Weighted marked reduced moment measures and \textit{K}-functions}

    Finally, by letting $\tilde\rho=\rho$ and $\tilde t((x,m),(x_1,m_1),\ldots,(x_{n-1},m_{n-1}))$ be given by the product of
    \begin{align}
    \label{e:ProductTestFun}
    \tilde t(m,m_1,\ldots,m_{n-1}) =& 
    t(m,m_1,\ldots,m_{n-1})
    \frac{
    \mathbf{1}\{ m \in E\}
    }{\nu(E)}
    \prod_{i=1}^{n-1}
    \frac{\mathbf{1}\{m_i\in E_i\}}{\nu(E_i)}
    ,
    \\
    \tilde w(x,x_1,\ldots,x_{n-1}) =& 
    w(x,x_1,\ldots,x_{n-1})
    \prod_{i=1}^{n-1}
    \mathbf{1}\{x_i \in (x + C_i)\}
    ,\nonumber
    \end{align}
     for $E\in \mathcal{B(M)}$, $\nu(E)>0$, and $C_i \times  E_i \in \mathcal{B}(\mathbb{R}^d) \times \mathcal{B(M)}= \mathcal{B}(\mathbb{R}^d \times \mathcal{M})$,  $\nu(E_i)>0$, $i=1,\ldots,n-1$, 
     we obtain an unbiased estimator $\hat{\mathcal{K}}_t^{ E \times_{i = 1}^{n-1}  E_i}(C_1 \times \cdots \times C_{n-1})
  = G(r,Y\cap W\times 
  \mathcal{M}
  )$ of the $t$-weighted marked $n$-th order reduced moment measure of \citet{ghorbani2021functional}, 
\begin{align}
\label{eq:tweight0}
&\mathcal{K}_t^{ E \times_{i = 1}^{n-1}  E_i}(C_1 \times \cdots \times C_{n-1})
  =
  \\
  =&
  \frac{1}{\vert W \vert \nu(E) \prod_{i=1}^{n-1}\nu(E_i)}
  \mathbb{E}\Bigg[\sum_{(x,m)\in Y \cap W \times E}  \sum_{(x_1,m_1), \ldots, (x_{n-1},m_{n-1}) \in Y \setminus \{ (x,m) \}}^{\neq} \frac{t(m,m_1,\ldots,m_{n-1})}{\rho(x,m)}
  \times
  \nonumber
  \\
  & \times 
  \prod_{i=1}^{n-1}\frac{\mathbf{1}\{x_i-x \in C_i\}\mathbf{1}\{ m_i \in  E_i\}}{\rho(x_i,m_i)}
  \Bigg]
  \nonumber
\end{align}
assuming that the edge correction function $w$ is such that unbiasedness holds. 
Examples of such $w$ include the minus sampling edge correction and the translational edge correction \citep{ghorbani2021functional}.
Note here that one just as well could have merged the scaled indicators in the expression for $\tilde t$ with $t$ so that $\tilde t=t$; \cite{ghorbani2021functional} included this mark set filtering to highlight that their summary statistic generalises previously proposed ones. 

\subsection{Local t-weighted marked $n$-th order  inhomogeneous K-function}

In this section, we provide the estimator corresponding to the local contributions of \eqref{eq:tweight0} and discuss its properties.


\begin{definition}
\label{def:LocalK}
Let $\tilde t$ be (up to indicator-scaling) as in \eqref{e:ProductTestFun} and consider 
\begin{equation}
    \begin{split}
&\hat{ \mathcal{K}}_t^{(x,m) \times_{i = 1}^{n-1}  E_i}(C_1 \times \cdots \times C_{n-1}) = L_n((x,m), Y\setminus\{(x,m)\}\cap W\times\mathcal{M}; \tilde t, \tilde  \rho)=
\\
=&  
\frac{1}{\tilde{\rho}(x,m)\nu(E) \prod_{i=1}^{n-1}\nu( E_i)}
\sum_{(x_1,m_1), \ldots, (x_{n-1},m_{n-1}) \in Y \setminus \{ (x,m) \}\cap W\times\mathcal{M}}^{\neq}  w(x,x_1,\ldots,x_{n-1}) 
\times
\\
  &\times  
  t(m,m_1,\ldots,m_{n-1})
  \prod_{i=1}^{n-1}\frac{\mathbf{1}\{x_i-x \in C_i\}\mathbf{1}\{ m_i \in E_i\} }{\tilde{\rho}(x_i, m_i)},
  \qquad
  (x,m)\in Y\cap W\times\mathcal{M},
  \label{eq:lo2}   
    \end{split}
\end{equation}
for some suitable edge correction $w$ in \eqref{e:ProductTestFun}, 
$W\in\mathcal{B}(\mathbb{R}^d)$, $E\in\mathcal{B}(\mathcal{M})$, $\nu(E)>0$, and $C_i \times  E_i \in \mathcal{B}(\mathbb{R}^d \times \mathcal{M})$, $\nu(E_i)>0$, $i=1,\ldots,n-1$.
We refer to 
$\hat{ \mathcal{K}}_t^{(x,m) \times_{i = 1}^{n-1}  E_i}(r)=\hat{ \mathcal{K}}_t^{(x,m) \times_{i = 1}^{n-1}  E_i}(b[0,r]^{n-1})$, $r\geq0$, 
as a \textit{local t-weighted marked $n$-th order  inhomogeneous $K$-function}. In particular, $\hat{\mathcal{K}}_{t,n}^{(x,m)}(r)=\hat{\mathcal{K}}_t^{(x,m) \times \mathcal{M}^{n-1}}(r)$ does not perform any explicit mark set filtering.

\end{definition}

Note first that when there is a common mark distribution which coincides with the reference measure on $\mathcal{M}$, setting $\tilde\rho=\rho$ we, for instance, obtain
\begin{align*}
\hat{\mathcal{K}}_{t,n}^{(x,m)}(r)
=&
\sum_{(x_1,m_1), \ldots, (x_{n-1},m_{n-1}) \in Y \setminus \{ (x,m) \}\cap (b[x,r]\cap W)\times\mathcal{M}}^{\neq} 
\frac{t(m,m_1,\ldots,m_{n-1}) w(x,x_1,\ldots,x_{n-1})}
{\rho_g(x)\rho_g(x_1)\cdots\rho_g(x_{n-1})}
\end{align*}
since $\nu$ must be a probability measure here.

Regarding the distributional properties of \eqref{eq:lo2}, when $Y$ is $k$-MIRS,  Theorem \ref{thm:Local} tells us that the expectation is given by 
\begin{align*}
&
\frac{1}{\nu(E)}
\prod_{i=1}^{n-1}
    \frac{1}{\nu(E_i)}
\int_{\mathbb{R}^d}\cdots\int_{\mathbb{R}^d}
w(x,x_1+x,\ldots,x_{n-1}+x)
    \prod_{i=1}^{n-1}
    \mathbf{1}\{x_i \in (x + C_i)\cap(W-x)\}
    \times
  \\
  \times&
 \Bigg(
 \int_{E_1}\cdots\int_{E_{n-1}}
 t(m,m_1,\ldots,m_{n-1})
 \frac{f_{0,x_1,\ldots,x_{n-1}}(m,m_1, \ldots, m_{n-1})}{f_{0}(m) f_{x_1}(m_1)\cdots f_{x_{n-1}}(m_{n-1})}
 \nu(dm_1) \cdots \nu(dm_{n-1})
 \Bigg)
    \times
    \\
    \times&
 g_g^{(n)}(0, x_1, \ldots, x_{n-1})
 \text{d}x_1 \cdots \text{d}x_{n-1}
  .
\end{align*}
In particular, under independent marking the mark related integral within brackets reduces to 
$\int_{E_1}\cdots\int_{E_{n-1}}
t(m,m_1,\ldots,m_{n-1})
\nu(dm_1) \cdots \nu(dm_{n-1})$, whereby \eqref{eq:lo2} is given by the product of this term and a term measuring intensity reweighted spatial interaction.

\subsubsection{Test functions for FMPPs}
Turning to the FMPP case, by choosing different test functions $t(\cdot)$ for the functional marks, we may extract different features. 
We here focus on pairwise interactions, i.e.\
$n = 2$. 

The test function $t$ is intended to reflect similarities between functions. Hence, a natural starting point would be a metric $t(f_1,f_2)= d(f_1,f_2)$ on the function space $\mathcal{M}$, which does not necessarily need to be the underlying assumed metric on $\mathcal{M}$. 
The first candidate that comes to mind is an $L_p$-distance: 
\begin{equation}
    t(f_1,f_2)=\Bigg(\int_a^b|f_1(t)-f_2(t)|^p\text{d}t\Bigg)^{1/p},
    \quad 1\leq p\leq\infty,
    \label{eq:testl}
\end{equation}
where $p=\infty$ represents the supremum metric. 
For any choice of $p$ in \eqref{eq:testl}, similarity between functions implies a small value of the test function. Other tentative functions are semi-metrics based on the $L_p$ distance between the $s$-th derivatives of 
the 
functions, for different combinations of $p$ and $s$, with the $L_1$ and $L_2$ distances being particular cases, and semi-metrics based on functional principal component analysis.

A further alternative 
is the functional marked counterpart of the
test function for the classical variogram, given by
\begin{equation}
    t(f_1,f_2)=\int_a^b(f_1(t)-\Bar{F}(t))(f_2(t)-\Bar{F}(t))\text{d}t,
    \label{eq:testf}
\end{equation}
with $\Bar{F}(t)=(1/n)\sum_{i=1}^n f_i(t)$ being the average functional mark at time $t$ for the observed functional part of the point pattern; such averaging is motivated by the assumption of a common mark distribution.


\section{Local test for random labelling} \label{sec:test}

Simple hypotheses for spatial point patterns, such as Complete Spatial Randomness, are commonly tested using an estimator of a global summary statistic, e.g., Ripley's $K$-function. In this context, one typically resorts to Monte Carlo testing.
The first step is then to generate $Q$ simulations under the null hypothesis, and to
estimate the chosen summary statistic for both the observed pattern and the simulations. 
In order to study whether there is \textit{random labelling} in a (functional) marked point process,  the simulations are obtained by permuting the (functional) marks, that is, randomly assigning them to the spatial points of the ground pattern, which are kept fixed. Then, the chosen summary statistic is estimated for each of these permutations and global envelopes at a given nominal level are generated based on them. The result of the test can be assessed graphically: if the summary statistic estimate for the observed pattern 
exits 
the envelopes, we proceed with the assumption that the 
underlying FMPP is   
not randomly labelled.
Furthermore, it is possible to calculate a 
$p$-value based on the position of the observed summary statistic within the $q$th envelopes, following \cite{myllymaki2017global}. 
We know, however, that the conclusion drawn from the application of the above-mentioned global test 
pertains 
to the whole analysed process, indicating whether all the functional marks 
are randomly labelled or not.
Motivated by the will to further detect the specific points, and regions, where the functional marks really do depend on the other marked points, 
we propose a \textit{local test for random labelling}. The main idea is to run a global envelope test on each point of the analysed pattern by means of the previously proposed \textit{local t-weighted marked inhomogeneous $K$-functions}, to draw different conclusions about the individual points, based on the obtained $p$-values. 
In Algorithm \ref{alg:1} we outline the proposed local test. Note that we alternatively may use sampling without replacement in step \ref{step:resample} of Algorithm \ref{alg:1}. Moreover, if convinced that multiple testing issues are present here, one may adjust the type I error probability $\alpha$ by using e.g.~the Holm-Bonferroni method.
\begin{algorithm}
\caption{Local test of random labelling}\label{alg:1}
  \begin{algorithmic}[1]
  \State Set a fixed nominal value $\alpha$ for type I error;
   \State Consider a (functional) marked point pattern $\mathbf{x}=\{(x_j,m_j)\}_{j=1}^k$, $k\geq1$; 
   \State Set a number of simulations, $Q\geq1$;
   \For{each $q = 1, \ldots, Q$:}
        \State\label{step:resample} Randomly sample $k$ (functional) marks, with replacement, from the original $k$ ones; 
        \State Denote the resulting point pattern by  $\mathbf{x}_q=\{(x_j,m_j^q)\}_{j=1}^k$;  
         \EndFor    
      \For{each 
      $j=1,\ldots,k$,}
      \State Compute $L_n^{(j,q)} = \{\hat{\mathcal{K}}_t^{(x_j,m_j^q) \times_{i = 1}^{n-1} E_i}(r; \mathbf{x}_q)\}_{r\in[0,r_{max}]}$ for all $q=1,\ldots,Q$;
    
 \State 
    Apply global envelope testing, 
    using the functions $L_n^{(j,q)}$, $q=1,\ldots,Q$, to generate the envelopes;
    \State Obtain a $p$-value $p_j$ from the test;
    \State\label{setp:reject} Reject the null hypothesis for the $j^{th}$ point 
    if $p_j \leq \alpha$.
    \EndFor 
    
    \end{algorithmic}
    \end{algorithm}


\section{Motivating example and simulation study}\label{sec:sims} 

This section is dedicated to simulation studies to assess the performance of our proposed local test. First, section \ref{sec:extest} provides a motivating example of the use of such a test, by means of simulated data resembling seismic events, which in turn have motivated this work. In particular, this means simulating the functional marks as seismic waveforms, following the typical abrupt change in variance of the signal in correspondence with the arrivals of the first P- and S-waves.
Then, section \ref{sec:ex_sims} presents an extensive simulation study, showing diverse and more general settings. Specifically, we assess the performance of the test by summarising the results in terms of classification rates.

\subsection{The need for a local test} \label{sec:extest}


We simulate a homogeneous spatial point pattern with 250 points on the unit square, $W=[0,1]\times[0,1]$, which represents the ground pattern. For each ground point $x_i$, we simulate a functional mark of the from
	\begin{align*}
		f_i(t) =& y(t) = \mu(t)+\epsilon(t), 
		\qquad t\in\mathcal{T}= [0,1], 
		\\
		\epsilon(t) \sim & N(0,\sigma(t)^2),
		\\ 
		\sigma(t)^2 =& 0.2+7.5\mathbf{1}\{t>0.4\}-5\mathbf{1}\{t>0.6\}, 
	\end{align*}
	where the mean signal $\mu(t)$ is taken to be zero.
The spatial 
ground point pattern
and the corresponding waveform for a given point are shown in Figure \ref{fig:sim1}(a)-(b). 
Since the marks/waveforms are simulated from the same 
model, and independently of each other and the spatial locations of the points, we see that such a process is indeed random labelled.

\begin{figure}[htb]
	\centering
\includegraphics[width=.57\textwidth]{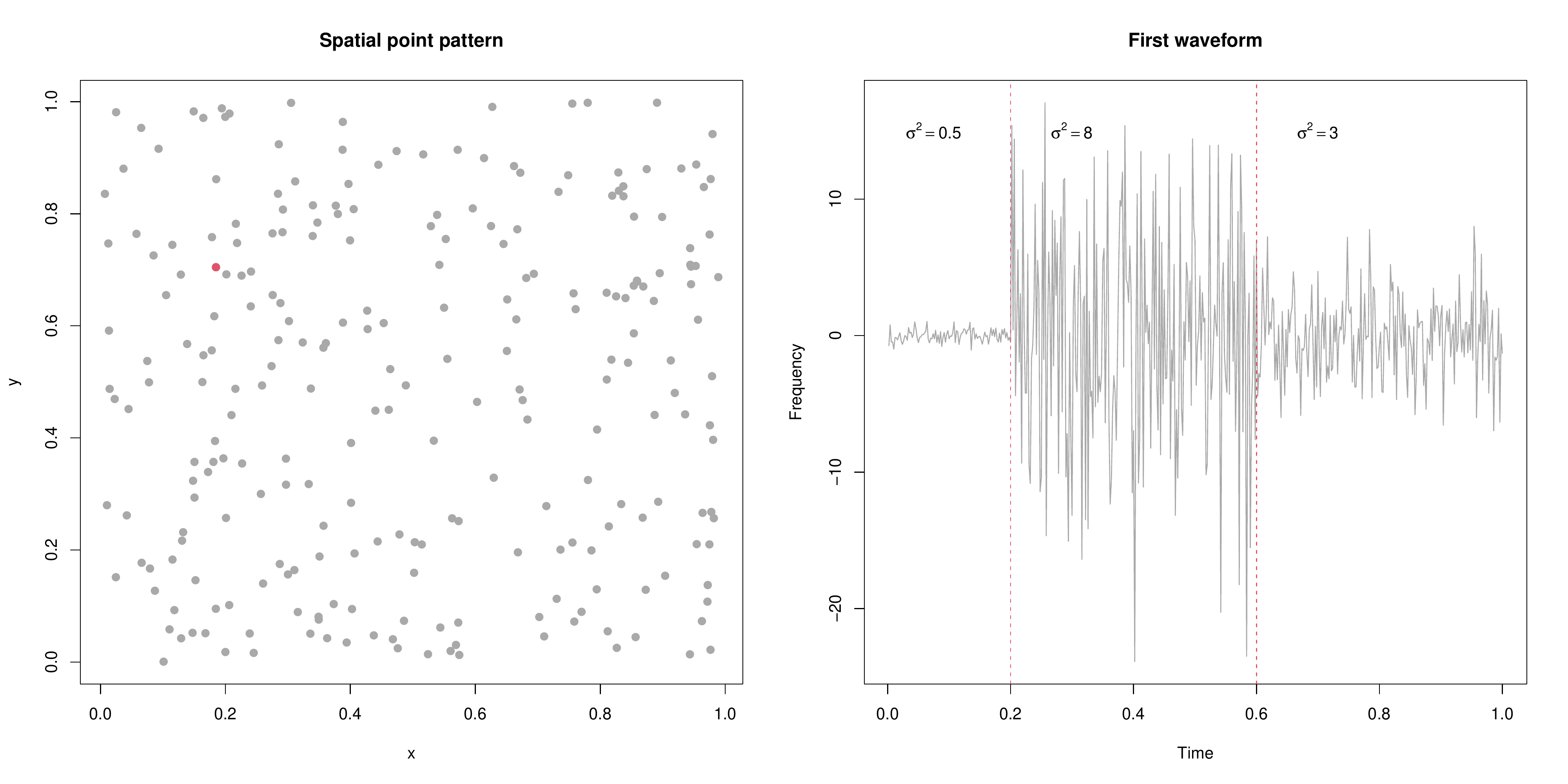}
\includegraphics[width=.42\textwidth]{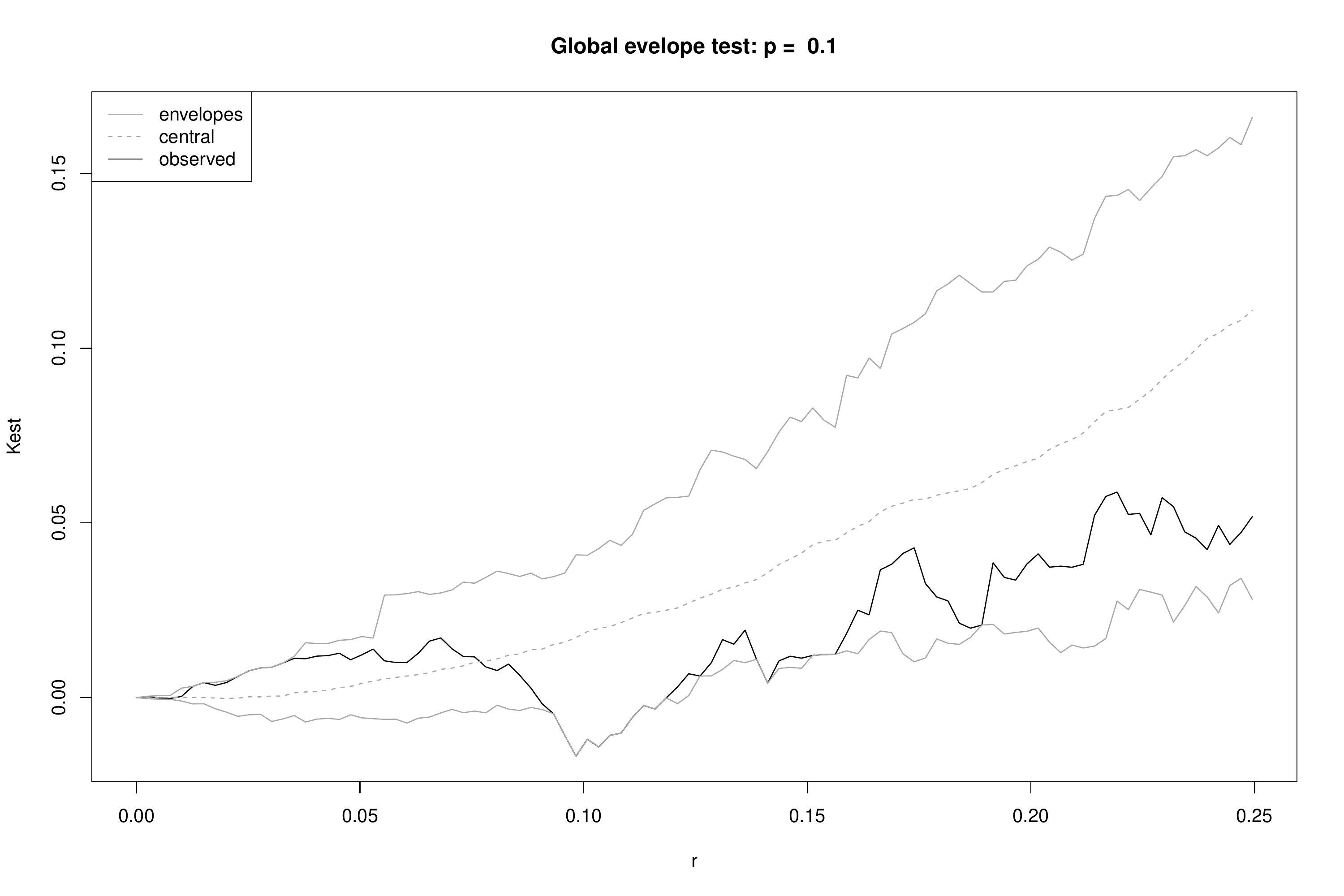}
	\caption{Simulated earthquake locations (a);  Simulated waveform  attached to the red point on the left (b); Result of global test (c).}
	\label{fig:sim1}
\end{figure}

Having generated the data, we first run a \textit{global envelope test for random labelling}, by randomly permuting the simulated waveforms, i.e.~the functional marks, keeping the location of the points fixed. 
We run the test by means of the t-weighted marked $n$-th order  inhomogeneous $K$-function of \cite{ghorbani2021functional}, with $n=2$, making it a second order summary statistic, and $t$ given by the test function \eqref{eq:testf}, i.e.\ the functional marked counterpart of the test function for the classical variogram. 
As previously mentioned, we assume that there is a common mark distribution which coincides with the reference measure on the mark space so that the intensity function is estimated by the ground process intensity estimate.  To be as objective as possible, we do not use the homogeneous intensity estimator $\widehat\rho_g(\cdot)=Y_g(W)/ \vert  W \vert $ here but instead we use a kernel intensity estimator, as in practice it would be unknown to us whether the actual ground process is (in)homogeneous. 
We use a Gaussian kernel intensity estimator $\widehat\rho_g(\cdot)$, where we select the bandwidth, $h$, according to \cite{cronie2018non}. More specifically, we minimise the discrepancy between the area of the observation window and the sum of reciprocal estimated intensity values at the points of the point pattern, i.e.\ we minimise $CvL(h) = (|W| - \sum_i 1/\hat\rho_g(x_i;h))^2$, where the sum is taken over all the data points $x_i$ and $\hat\rho_g(x_i;h)$ is the kernel intensity estimate with bandwidth $h$, evaluated in 
$x_i$. 
Then, once the bandwidth has been selected, the intensity estimate is corrected for edge effects through global edge correction (the option \texttt{diggle=FALSE} in the \texttt{spatstat} function \texttt{density.ppp}), i.e.~dividing the estimate by the convolution of the Gaussian kernel with the window of observation \citep{diggle1985kernel}. 
Finally, for $w$ we use Ripley’s isotropic edge correction in the summary statistic to correct for edge effects.
We repeated the procedure $39$ times, obtaining the result depicted in Figure \ref{fig:sim1}(c). 
We stress that our approach seems to be robust with respect to the bandwidth specification, i.e.~the choice of bandwidth selection approach plays a minor role for the final result.



As evident from Figure \ref{fig:sim1}(c), the observed summary statistic completely lies within the envelopes, and this confirms the expected result of lack of spatial dependence/structure of the functional marks. 
This result is further corroborated by the non-significant $p$-value, equal to $0.1$.

\subsubsection{Simulating spatially dependent functional marks}

To make the functional marks spatially dependent, we then superimpose a homogeneous spatial point pattern with $50$ points, generated in the $[0,0.5] \times [0,0.5]$ square, i.e. the bottom left region of the entire study region $W$. For these additional points, we generate different functional marks than before, namely with the underlying trend $\mu(t)=10+6\sin(3\pi z_t)$. Consequently, we have simulated a FMPP with spatially varying functional marks, i.e. not random labelled. We therefore expect a global test of random labelling to confirm this.

\begin{figure}[htb]
	\centering
\includegraphics[width=.425\textwidth]{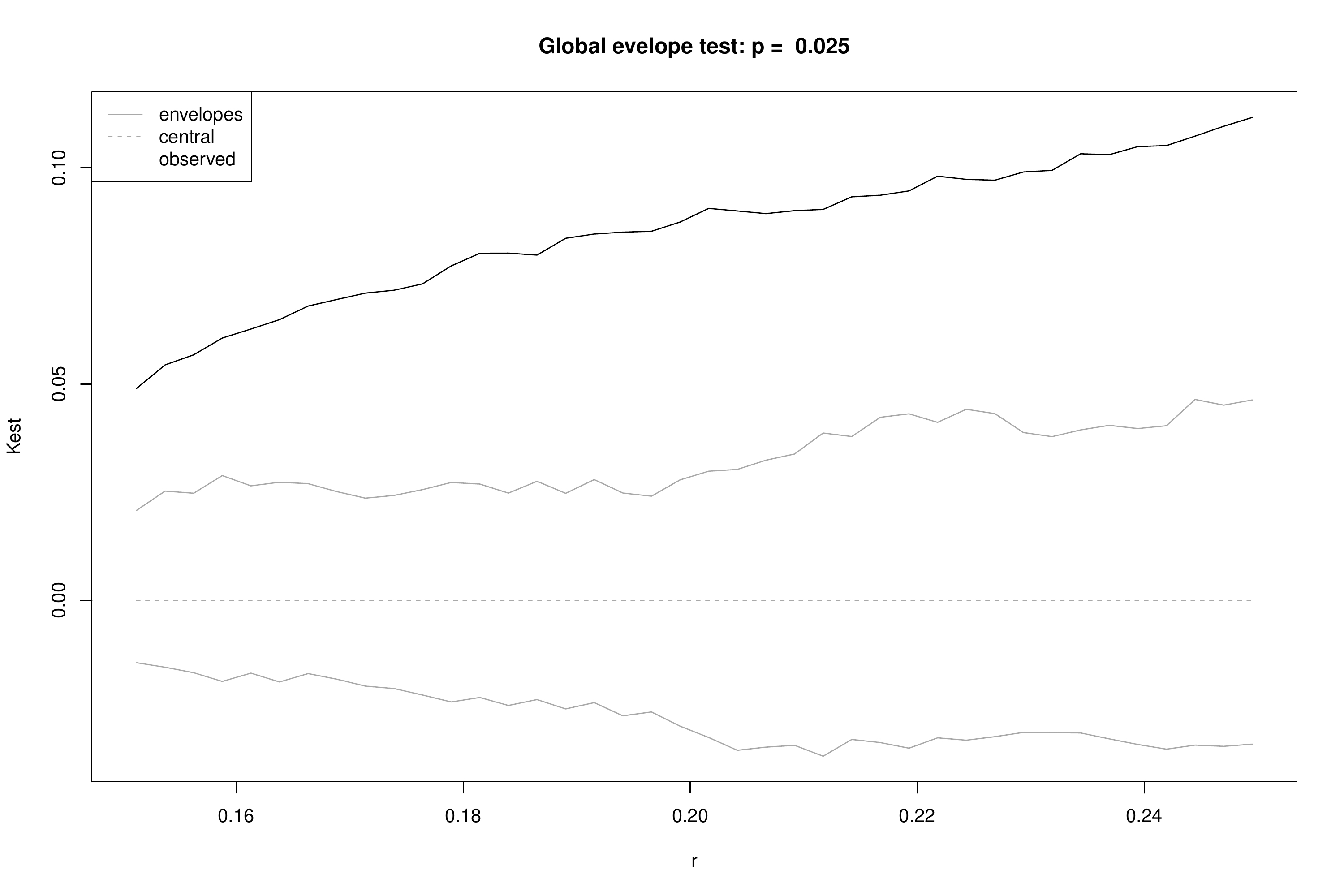}
\includegraphics[width=.275\textwidth]{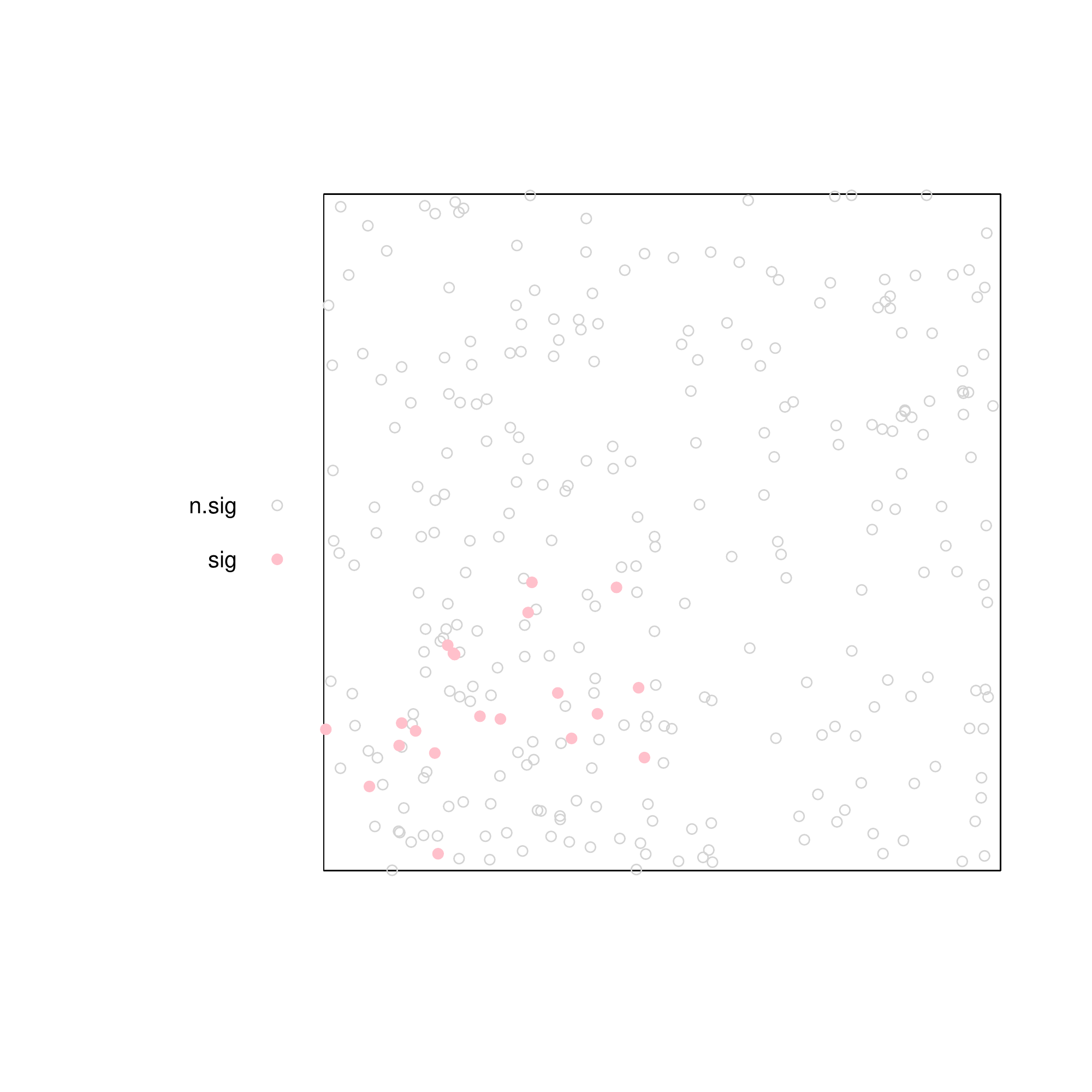}
	\caption{ 
	Result of global test for the spatially dependent simulated data (a);
	Output of the local test: significant points for which the hypothesis of random labelling is rejected are in pink (b).}
	\label{fig:5}
\end{figure}



We first run the same global test of random labelling as before. 
We use $Q = 39$ and obtain a global $p$-value of $0.025$. This, together with the observed $K$-function lying outside the envelopes (Figure \ref{fig:5}(a)), indicates the ability of the global test to correctly detect the spatial dependence of the functional marks.

We know, however, that this conclusion should not be drawn for each point of the pattern, if we consider local restrictions of it, but specifically for those in the vicinity of the $[0,0.5] \times [0,0.5]$ square. 
We therefore proceed by running our proposed local test, based on the proposed second order local $K$-function $\hat{\mathcal{K}}_{t,2}^{(x,m)}(r)$, $r\in[0,r_{max}]$, in Definition \eqref{def:LocalK}, 
with the same choice of test function $t(\cdot)$ and the same intensity estimation scheme as for the global one.
Figure \ref{fig:5}(b) depicts the points of the simulated point pattern, and it displays in pink those points for which the local test came out significant. Hence, this illustrates that the proposed local test is able to correctly identify some of the points, and consequently some parts of the region, where the hypothesis of random labelling does not hold locally. 
Note that a universally preferable option for  $r_{max}$ does not exist.
In this paper, it is set to $\min(x_W,y_W)/4$, where $x_W$ and $y_W$ represent the maximum width and height of the observation region $W$, respectively; note that this rule of thumb is supported by \cite{diggle:13}.
Indeed, changing the value of  $r_{max}$ has an impact on the final results, and we found that our choice provided the best compromise among the options.

\subsection{Extended simulation study}  \label{sec:ex_sims}

This section aims to study the proposed method's performance in terms of classification rates considering different scenarios, concerning both the ground processes and the functional marks' structures. 
To this end, we simulate under different such scenarios, to obtain a comprehensive understanding of the results of the local test in different settings.

In detail, we consider three types of ground process structures, all with an expected point count of 200: (1) a homogeneous Poisson process; (2) an inhomogeneous Poisson process with intensity function $\rho_g(x)=\rho_g(x_1,x_2)=\exp(3.5 + 3x_2)$, $x\in W$; (3) a Thomas process, with intensity of the Poisson process of cluster centres equal to 25, standard deviation of random displacement of a point from its cluster centre equal to 0.05, and mean number of points per cluster equal to 7.
They are all generated in $W$, i.e.\ the unit square, and will be referred to as the \textit{base patterns}. 
Then, we superimpose additional simulated patterns in the $[0,0.5] \times [0,0.5]$ square, coming from the same generating processes, but with an expected number of points of 50; hereby the expected total number of points on $[0,0.5] \times [0,0.5]$ is $50+200/4=100$ and on its complement it is $150$. These additional patterns will be referred to as \textit{feature patterns}.

As for the functional marks, we consider the time domain $\mathcal{T}=[0,10]$ and, practically, we sample each simulated mark function in 100 equally spaced time points in $\mathcal{T}$. 
We assume that each functional mark satisfies $f_i(t)=Z(x_i,t)$, where $x_i$ is the $i$th ground point and 
\begin{align}
\label{e:BaseModel}
Z(x,t) = \mu + \xi(x,t), \quad (x,t)\in W\times\mathcal{T}, 
\end{align}
for a zero-mean stationary Gaussian random field $\xi$ with covariance function $C(h,u)$; here $h$ and $u$ denote the spatial and the temporal lags, respectively. For the base patterns, we consider $\mu=5$ and a pure nugget effect model with covariance function $C(h,u)=\sigma^2\mathbf{1}\{h=0\}$, $\sigma^2=0.01$. In other words, each $f_i$ is random noise with mean 5 and variance 0.01 and all $f_i$'s are iid; see the grey curves in the bottom panel of Figure \ref{fig:sims}. 
For the feature patterns, we consider three different marking models:
\begin{enumerate}
    \item \label{itm:first} Shifted base model: We here let $\xi$ have the same form as in the base model but let $\mu=5.5$.
    \item \label{itm:second} Decreased variance base model: We here let $\xi$ have the same form as in the base model but let $\sigma^2=0.001$.
    \item \label{itm:third} Non-separable space-time model: We here let $\mu=5$ and consider a space isotropic 
    covariance function given by $C(h,u)= (\psi(u)+1)^{-\delta/2} \phi(h /\sqrt{\psi(u) +1})$. Here, $\phi$ is a normal mixture and the corresponding covariance function only depends on the distance 
    between two points, while  
    $\psi$ is a variogram model, which we choose according to a fractal Brownian motion with fractal dimension $\alpha = 1$; this is an intrinsically stationary isotropic variogram model. 
\end{enumerate}
We note that the first two of these scenarios represent independent but not identically distributed marks, whereas in the third scenario we additionally have that the marks are also dependent. 
A graphical representation of the scenarios comes in Figure \ref{fig:sims}. On the top panel, we display the three ground patterns. On the bottom ones, the corresponding functional marks are depicted.

\begin{figure}[htb]
	\centering
\subfloat{\includegraphics[width=.33\textwidth]{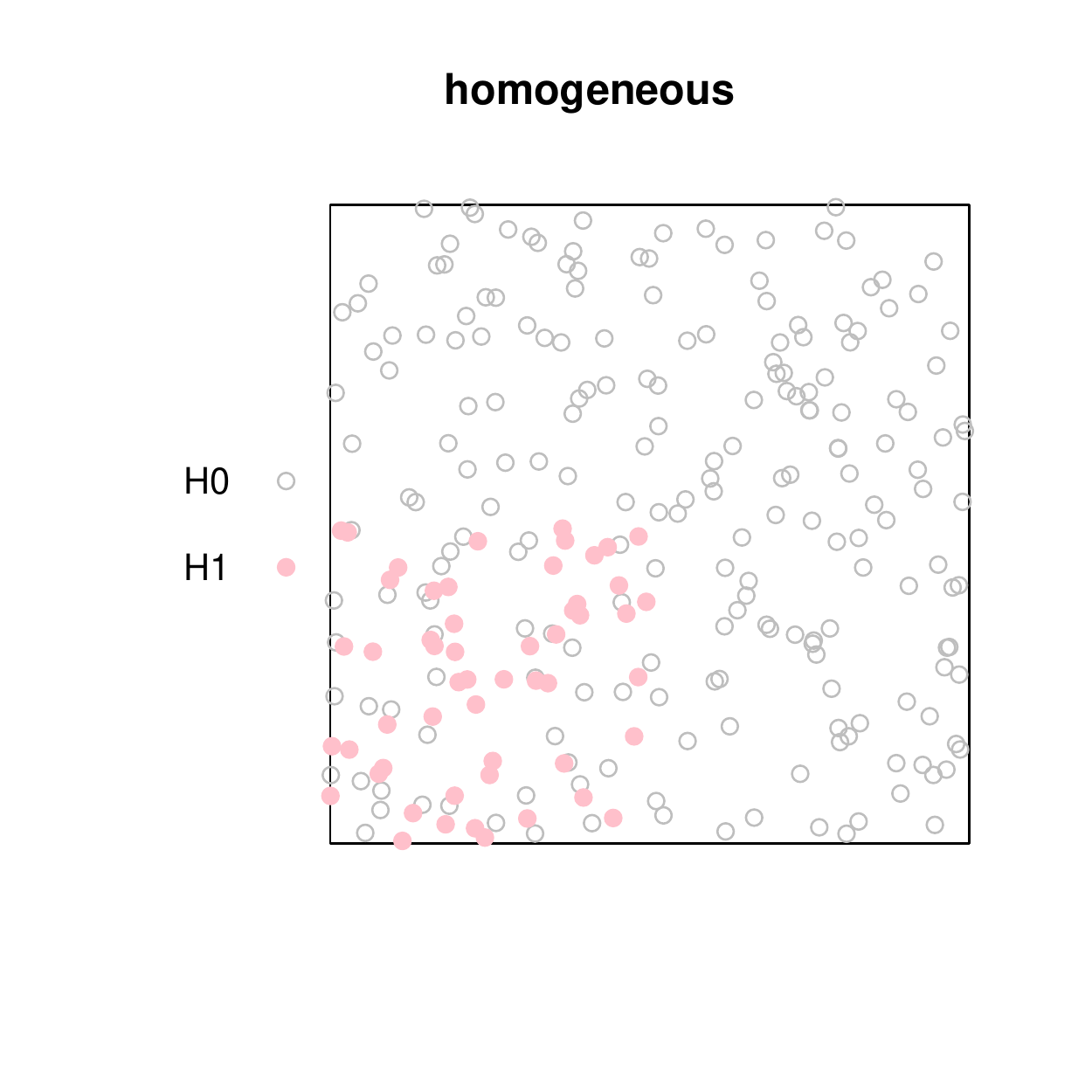}}
\subfloat{\includegraphics[width=.33\textwidth]{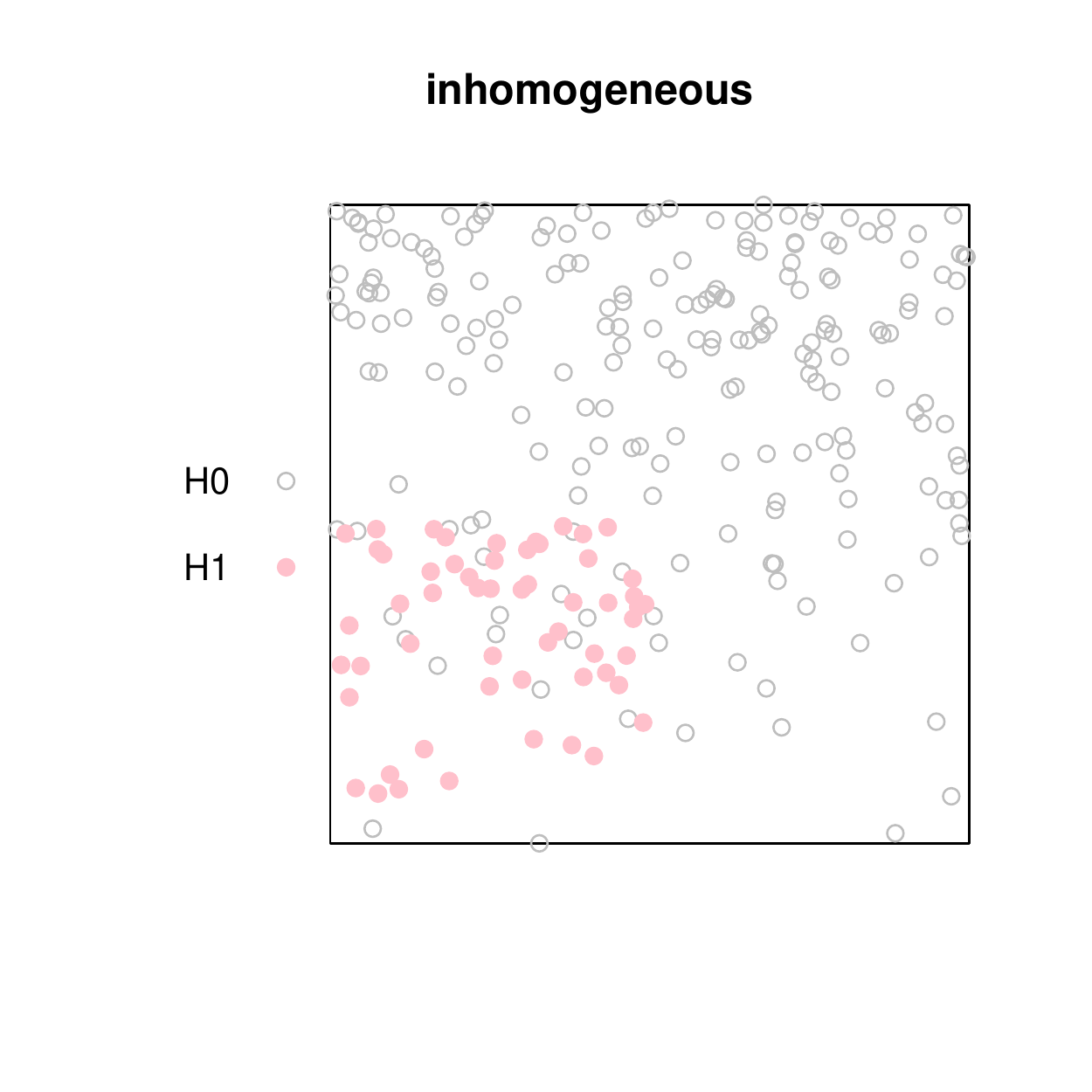}}
\subfloat{\includegraphics[width=.33\textwidth]{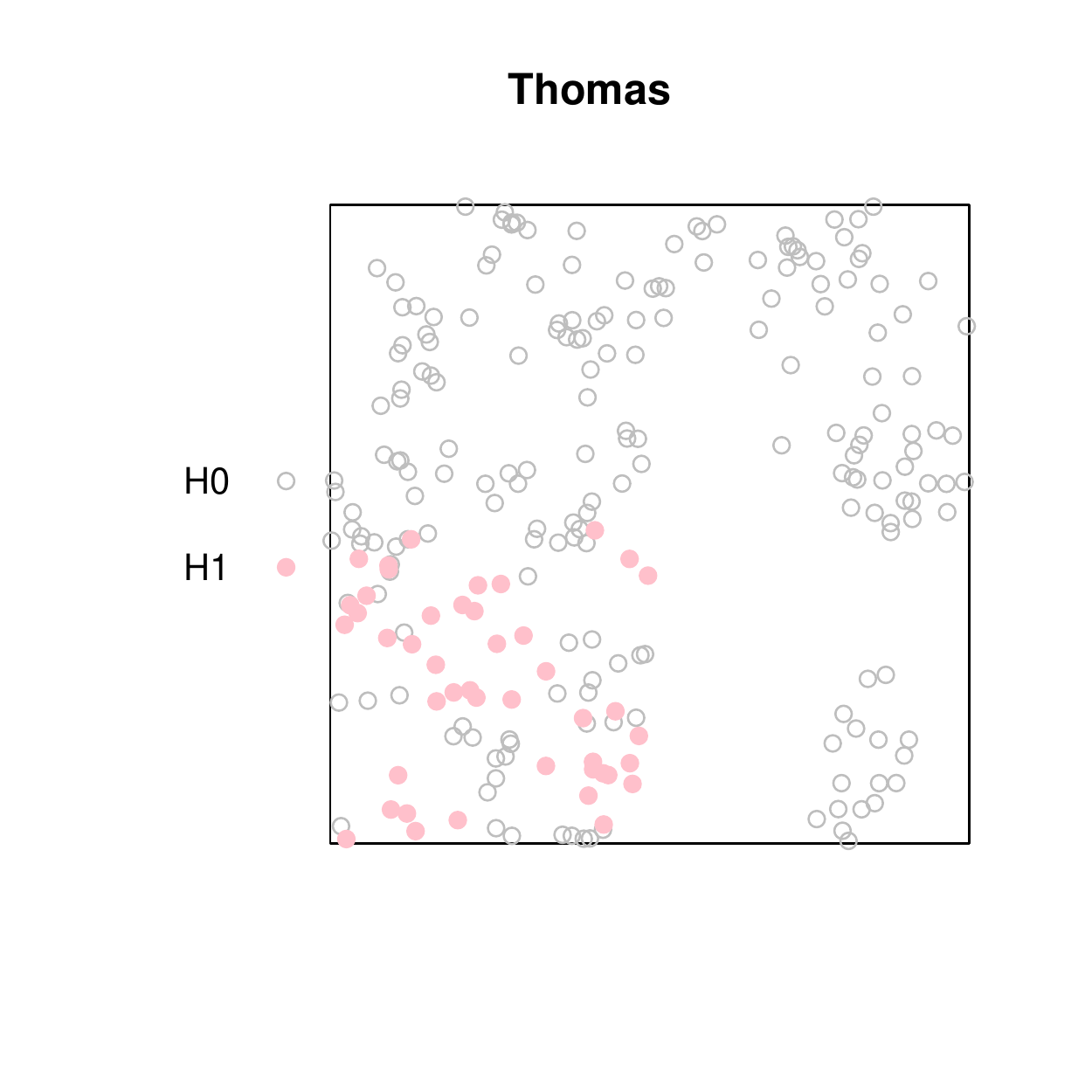}}\\
\vspace{-1.5cm}
\subfloat{\includegraphics[width=.33\textwidth]{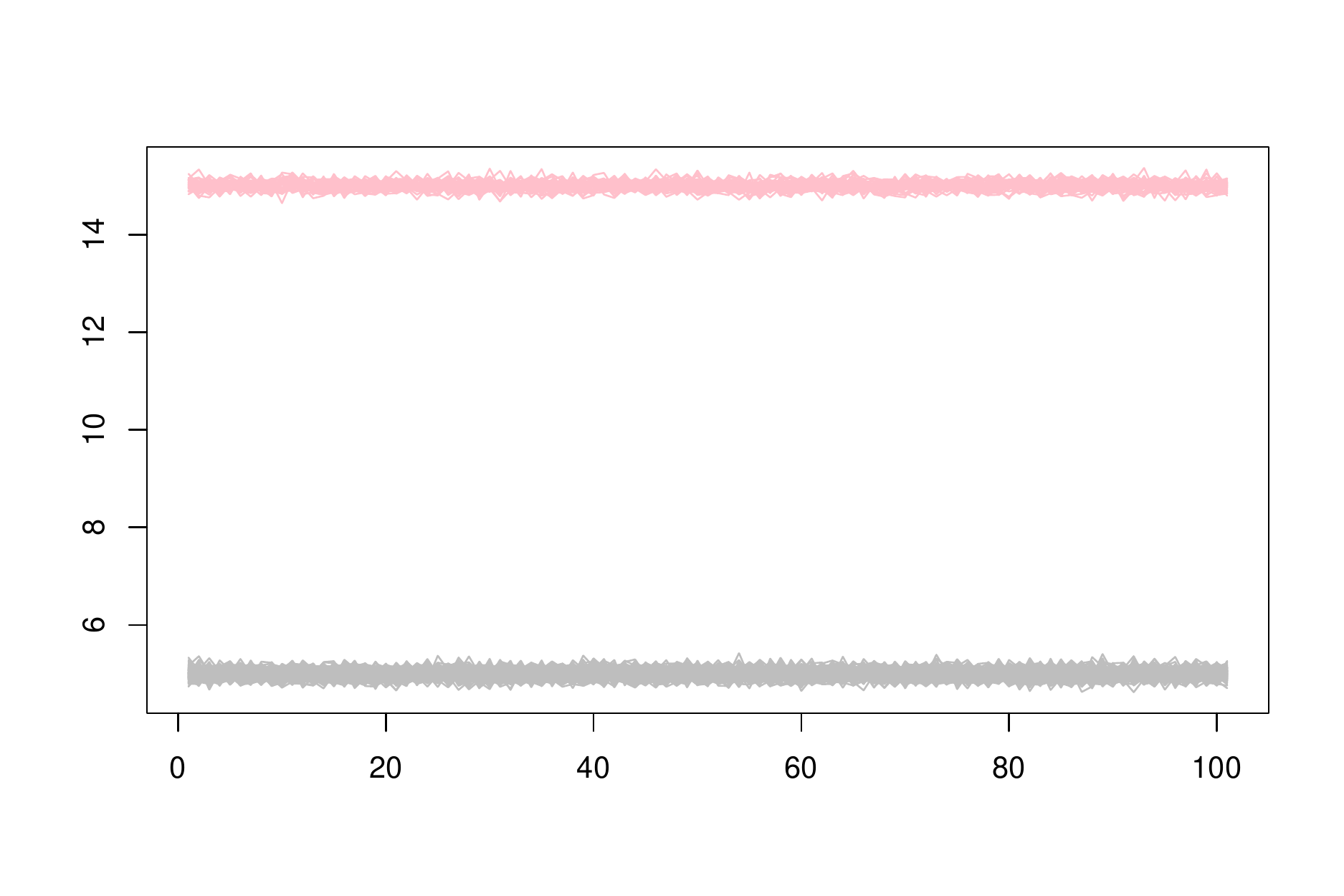}}
\subfloat{\includegraphics[width=.33\textwidth]{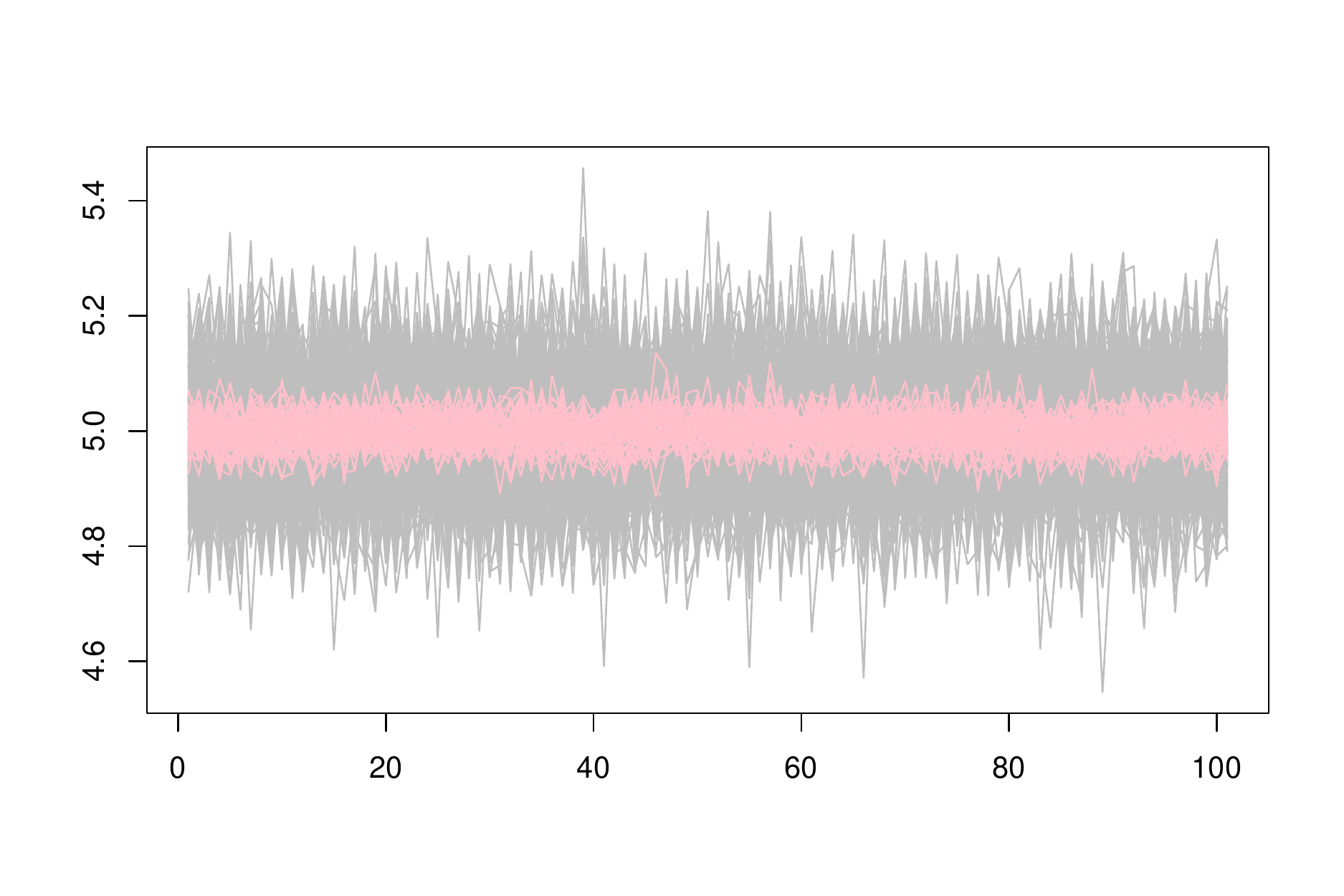}}
\subfloat{\includegraphics[width=.33\textwidth]{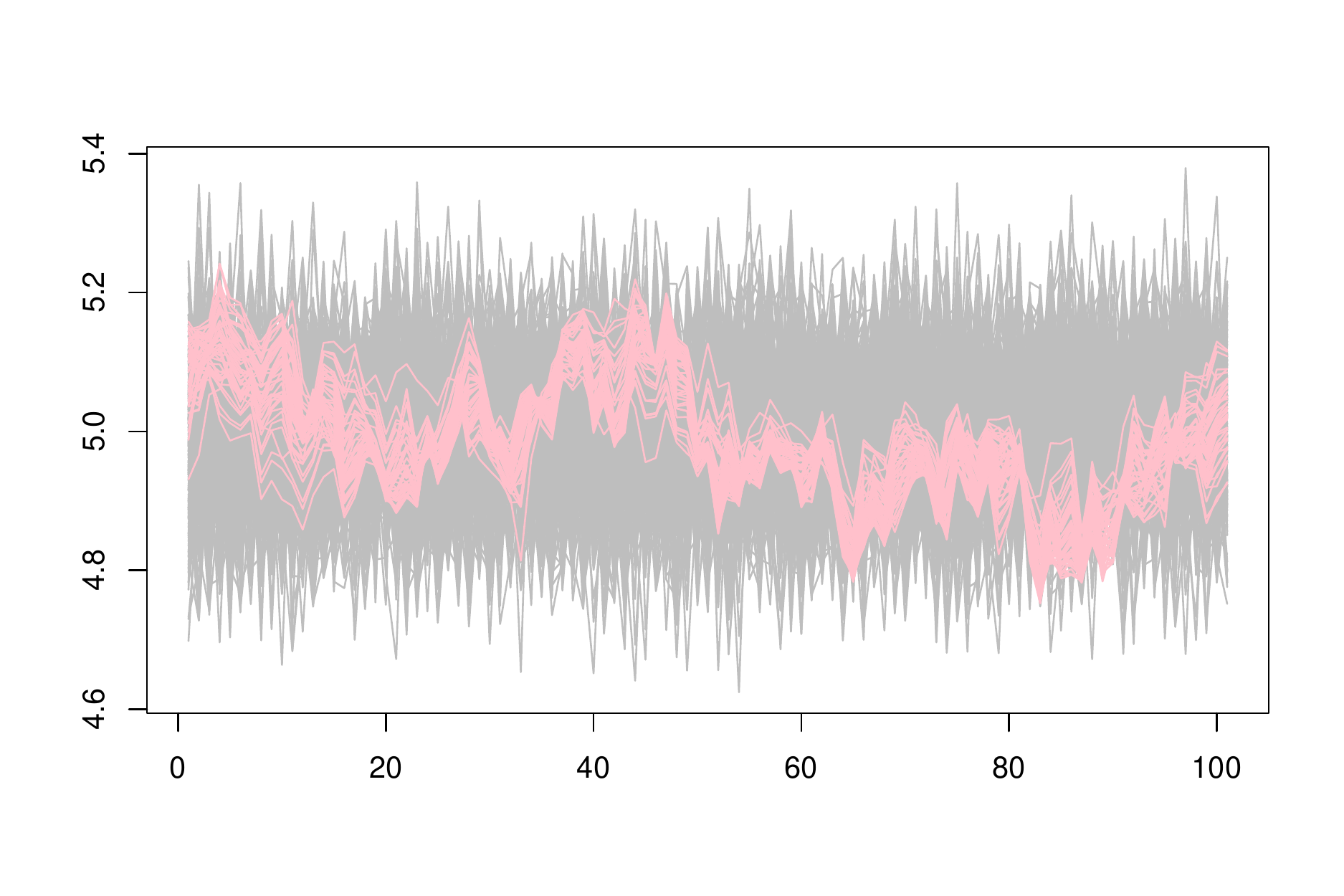}}
	\caption{Simulation scenarios. \textit{Top panels}: spatial ground patterns; \textit{Bottom panels}: functional marks of model \eqref{e:BaseModel} (in grey) and of the marking models in item \ref{itm:first}, \ref{itm:second}, and \ref{itm:third}, from left to right (in pink).}
	\label{fig:sims}
\end{figure}

We show the results of the local test in terms of true-positive rate (TPR), false-positive rate (FPR), and accuracy (ACC), averaging over 100 simulated point patterns in Table \ref{tab:res}.
The rates are defined as
\begin{equation*}
    TPR = \frac{\text{true positives}}{\text{positives}}, \quad FPR = \frac{\text{false negatives}}{\text{negatives}}, \quad ACC = \frac{\text{true positives and negatives}}{\text{positives and negatives}}.
\end{equation*}
We of course wish to have TPR and ACC close to 1 and FPR close to 0. 
 
As shown in Table \ref{tab:res}, the performance of the local test in terms of classification rates strongly depends on the difference in the functional marks. Specifically, changing only the mean of the underlying random field is not enough for properly identifying the points of the feature patterns. 
This sufficiently improves when changing the variance only, but the best result is obtained when the whole model is changed, that is, changing the correlation structure.
The effect of the type of ground pattern is less evident but still present. The inhomogeneous Poisson scenario reports the best classification rates, followed by the Thomas and homogeneous Poisson ones.

Finally, we found that the test function $t(\cdot)$ based on the $L_2$ distance in Equation  \eqref{eq:testl} gave the better results overall. To further explore how the choice of test function influences the test, we also compared to a test function incorporating a derivative function accounting for the shape of the functional marks. This yielded similar results but turned out to be more computationally demanding.

\begin{table}[htb]
\centering
\caption{Results of the local test averaged over 100 simulated point patterns with an expected point count of 250 each. 
}
\begin{tabular}{l|l|ccc}
 \toprule
Ground process&Marking model&TPR&FPR&ACC\\
  \midrule
Homogeneous Poisson&\eqref{itm:first}&0.112 &0.346 &0.583\\
Homogeneous Poisson &\eqref{itm:second}&0.583 &0.066 &0.820 \\
Homogeneous Poisson&\eqref{itm:third}&0.870 &0.024 &0.896 \\
\addlinespace
Inhomogeneous Poisson&\eqref{itm:first}&0.032 &0.585& 0.449\\
Inhomogeneous Poisson&\eqref{itm:second}&0.648& 0.084& 0.856 \\
Inhomogeneous Poisson&\eqref{itm:third}&0.895& 0.023& 0.932\\
\addlinespace
Thomas&\eqref{itm:first}&0.109 &0.394 &0.571\\
Thomas&\eqref{itm:second}&0.637 &0.088& 0.846\\
Thomas&\eqref{itm:third}&0.865 &0.025 &0.925\\
  \bottomrule
\end{tabular}
\label{tab:res}
\end{table}

\section{Real seismic data analysis}\label{sec:real}

We analyse data coming from the \textit{ISTANCE} dataset, presented in Section \ref{sec:data}.
More specifically, we analyse a sample dataset provided at \url{http://www.pi.ingv.it/instance/}. 
The dataset contains 10000 records of 300 events, together with the associated metadata. 
Figure 	\ref{fig:pre1} shows the earthquake locations and time occurrences.

\begin{figure}[htb]
	\centering
		\includegraphics[width=.35\textwidth]{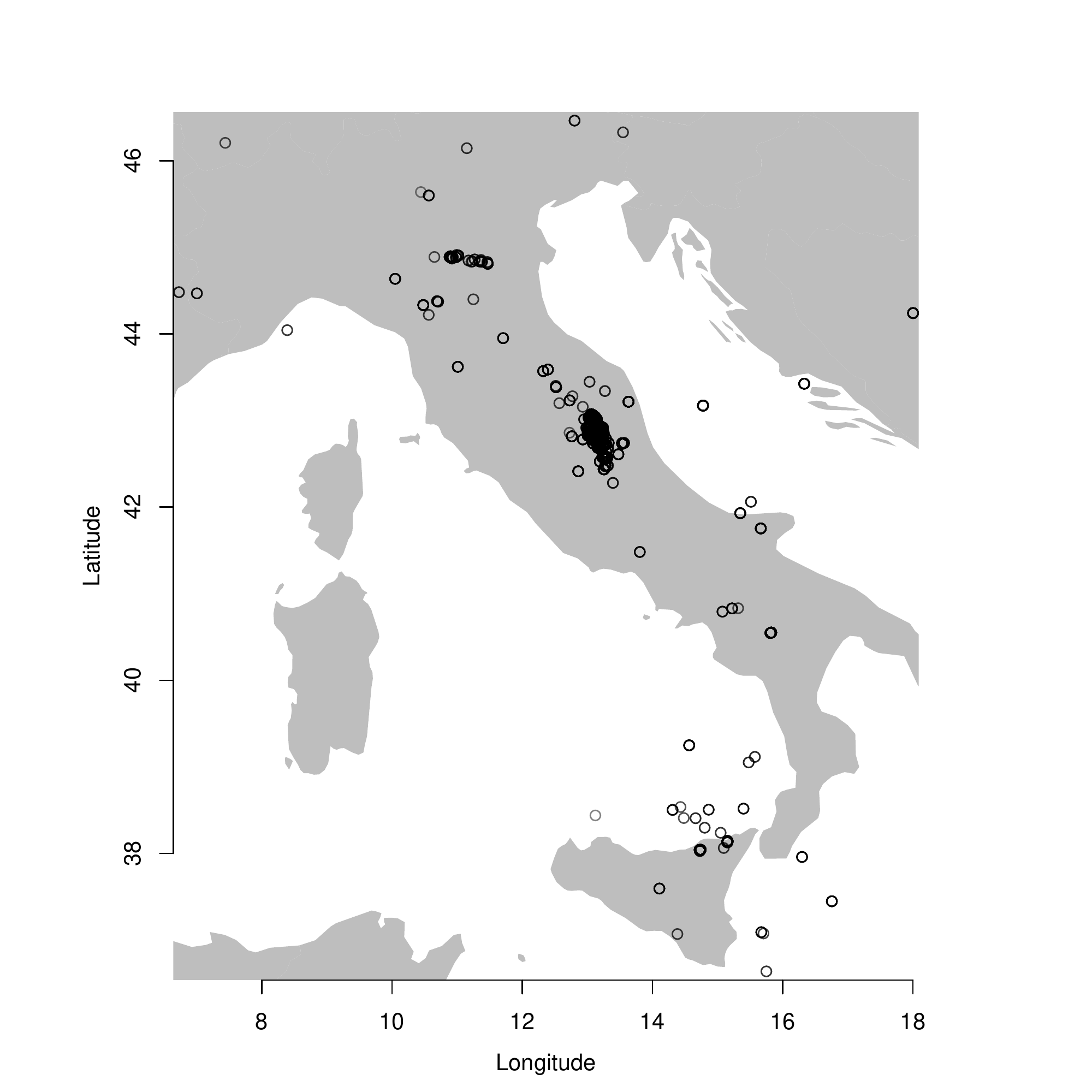}
			\includegraphics[width=.55\textwidth]{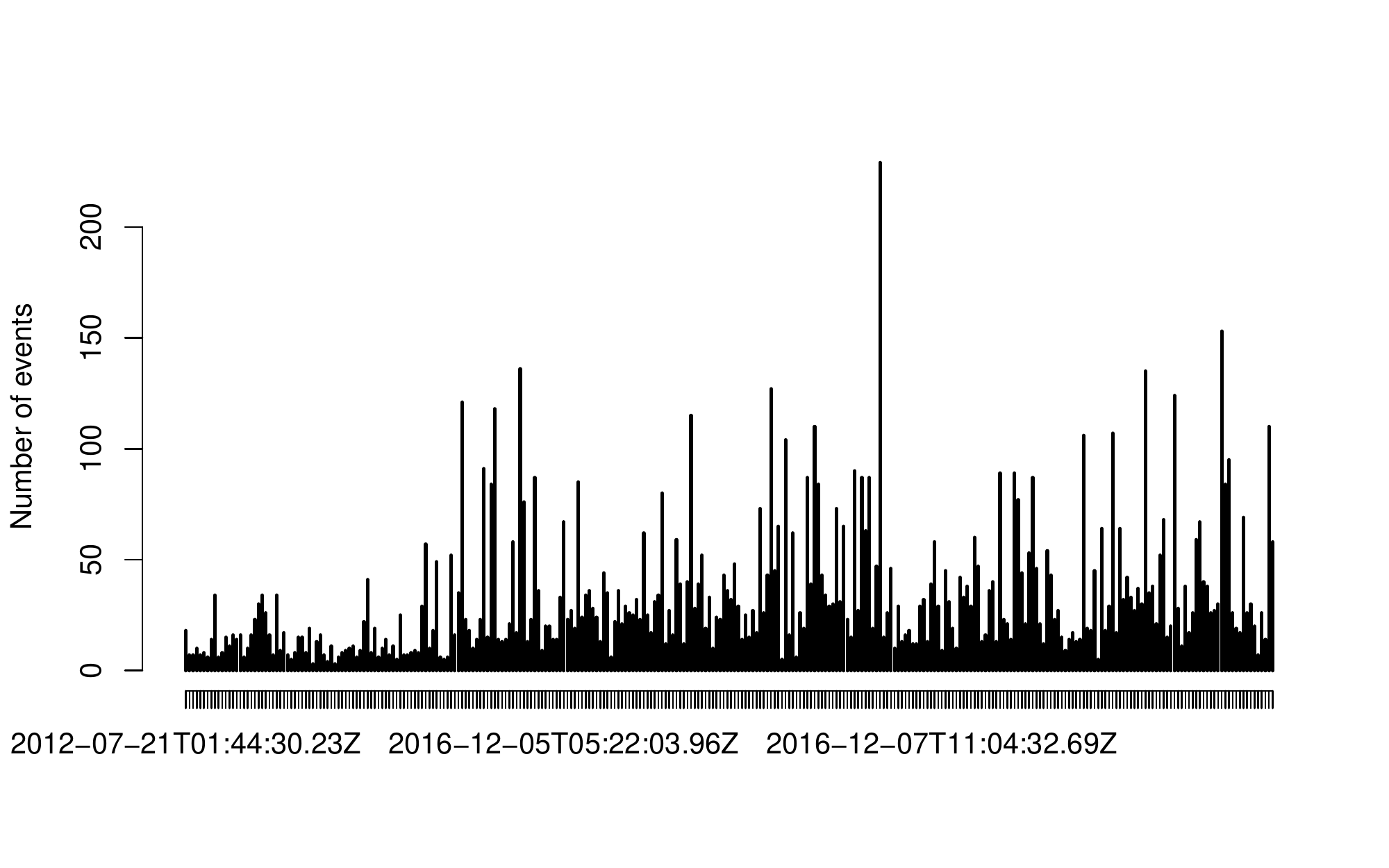}
	\caption{Earthquake locations and time occurrences}
	\label{fig:pre1}
\end{figure}

We first compute the proposed local $K$-function. The left panels of Figure \ref{fig:x cubed graph13} depict the estimated local summary statistics. In particular, the green lines represent the global statistics, while the grey ones represent the individual contributions. In blue we also represent the theoretical value. In the top-left panel, the $K$-function is based on a kernel intensity estimate whose bandwidth is selected by \cite{diggle:13}'s rule, while in the bottom-left panel the bandwidth is chosen as in \cite{cronie2018non}. We observe some relevant differences: while with \cite{diggle:13}'s rule we depict different local $K$-functions deviating from the global one, following \cite{cronie2018non}, we find a unique outlying local $K$-function. This may be explained by the fact that \cite{cronie2018non}'s approach tends to yield a bit too large bandwidths when large parts of the study region contain no points, while \cite{diggle:13}'s approach tends to yield too small bandwidths in general; see \cite{cronie2018non} for details. Note that by increasing the bandwidth we decrease the intensity estimate and, as a consequence, the summand denominators in \eqref{eq:lo2} are decreased. 
Therefore, we run the proposed local test of random labelling with both options for the bandwidth selection and, as expected, the differences observed in the computation of the local $K$-functions are reflected in the results of the test.

\begin{figure}[htb]
	\centering
	\subfloat{\includegraphics[width=.33\textwidth]{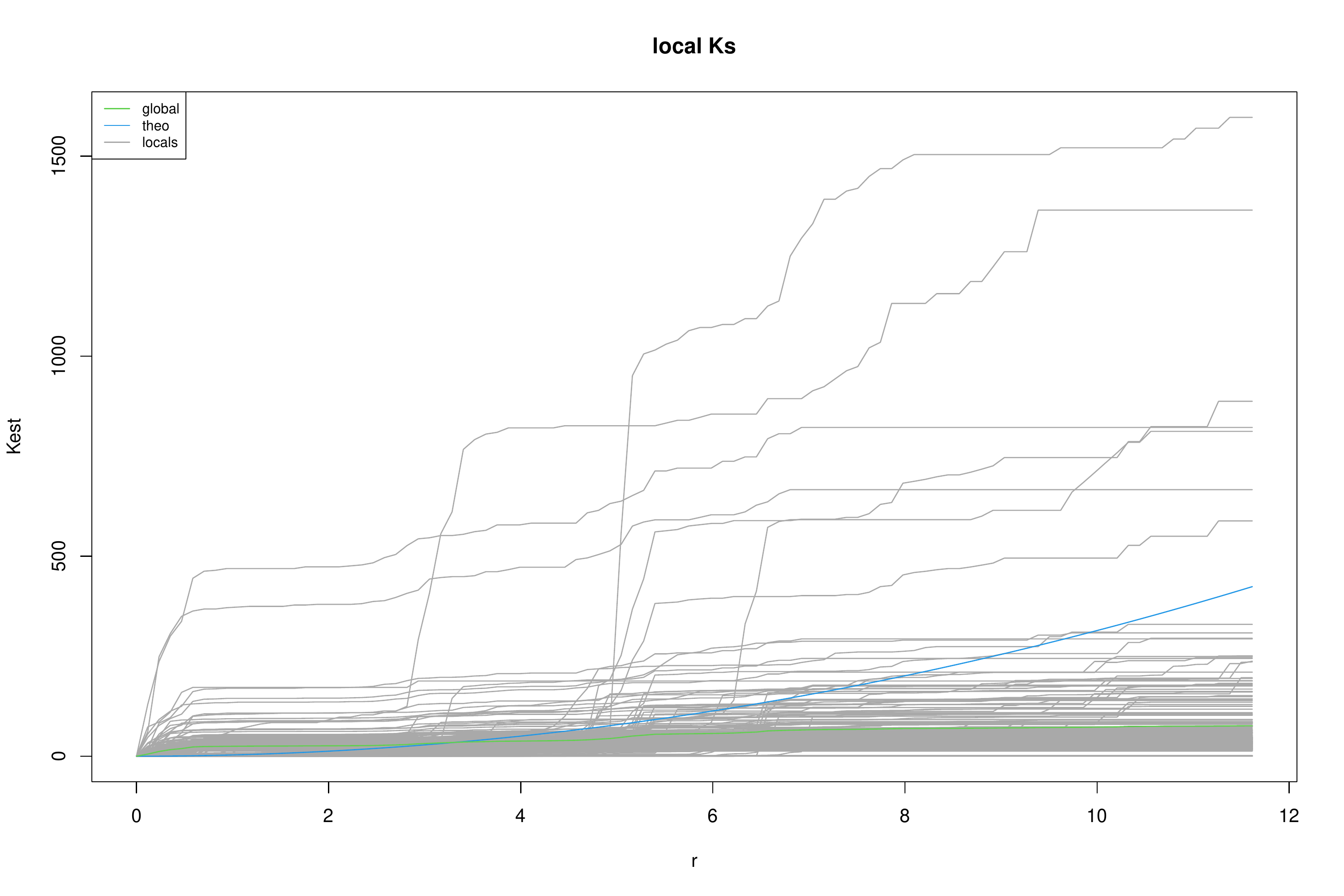}}
		\subfloat{\includegraphics[width=0.33\textwidth]{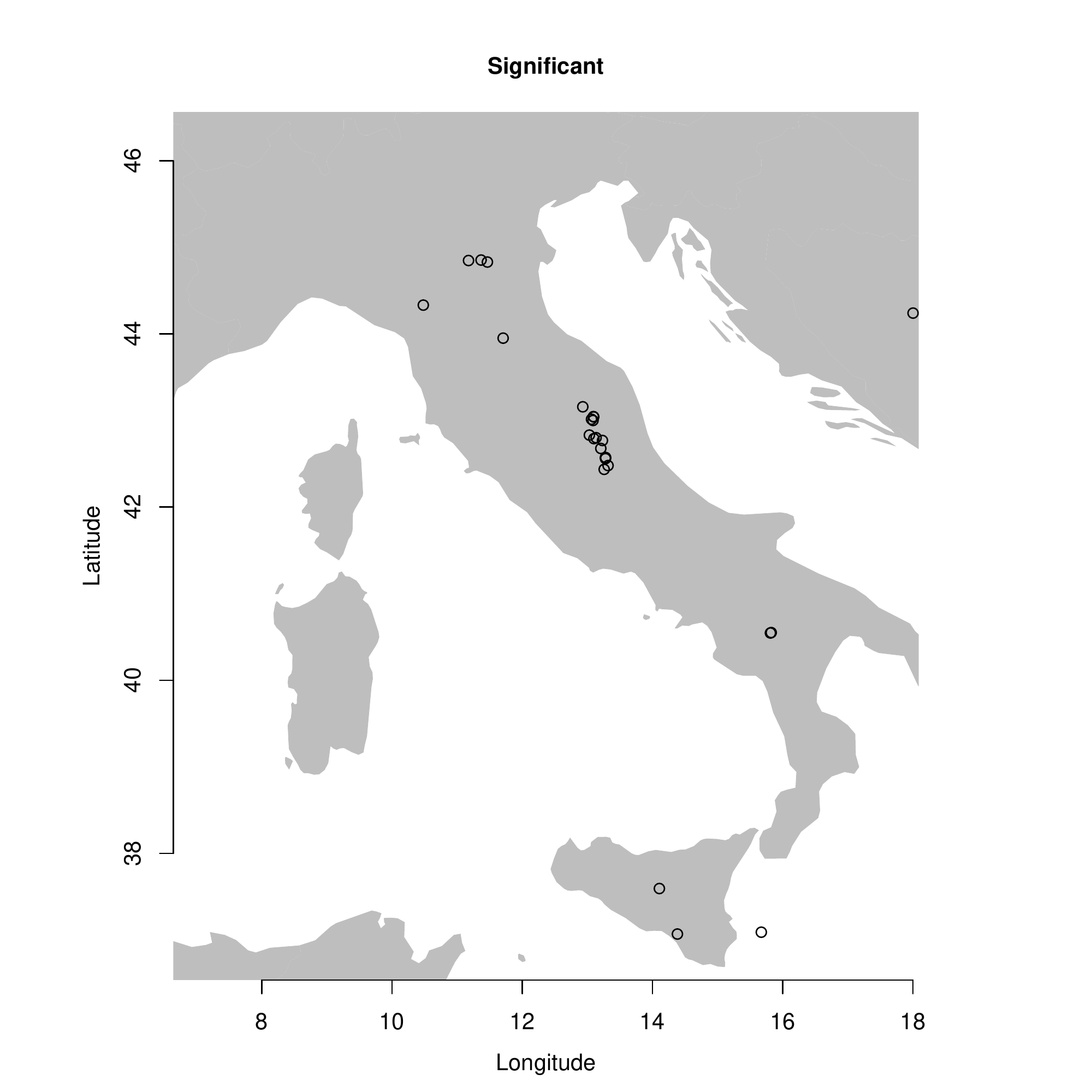}}
			\subfloat{\includegraphics[width=0.33\textwidth]{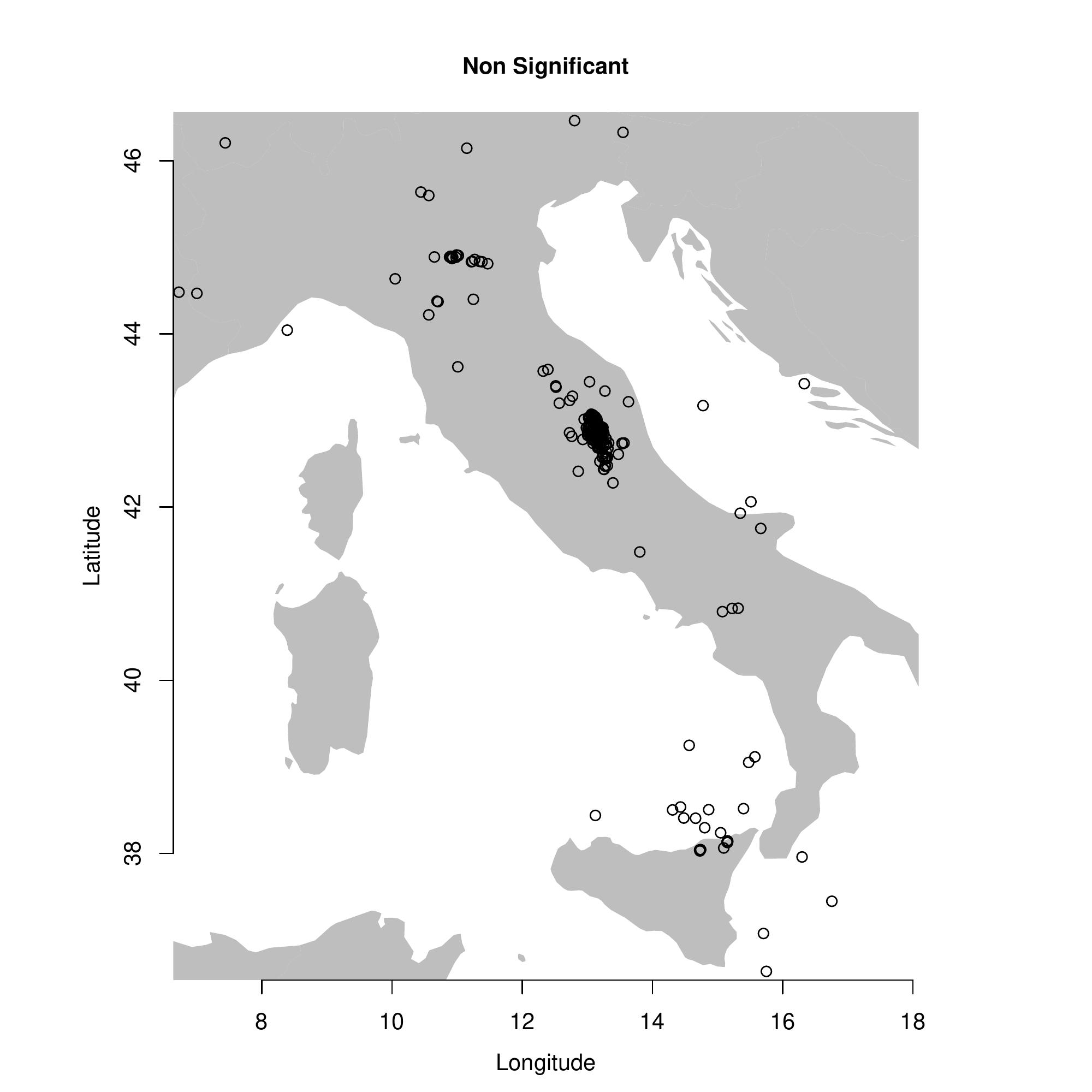}}
		\\
	\subfloat{\includegraphics[width=.33\textwidth]{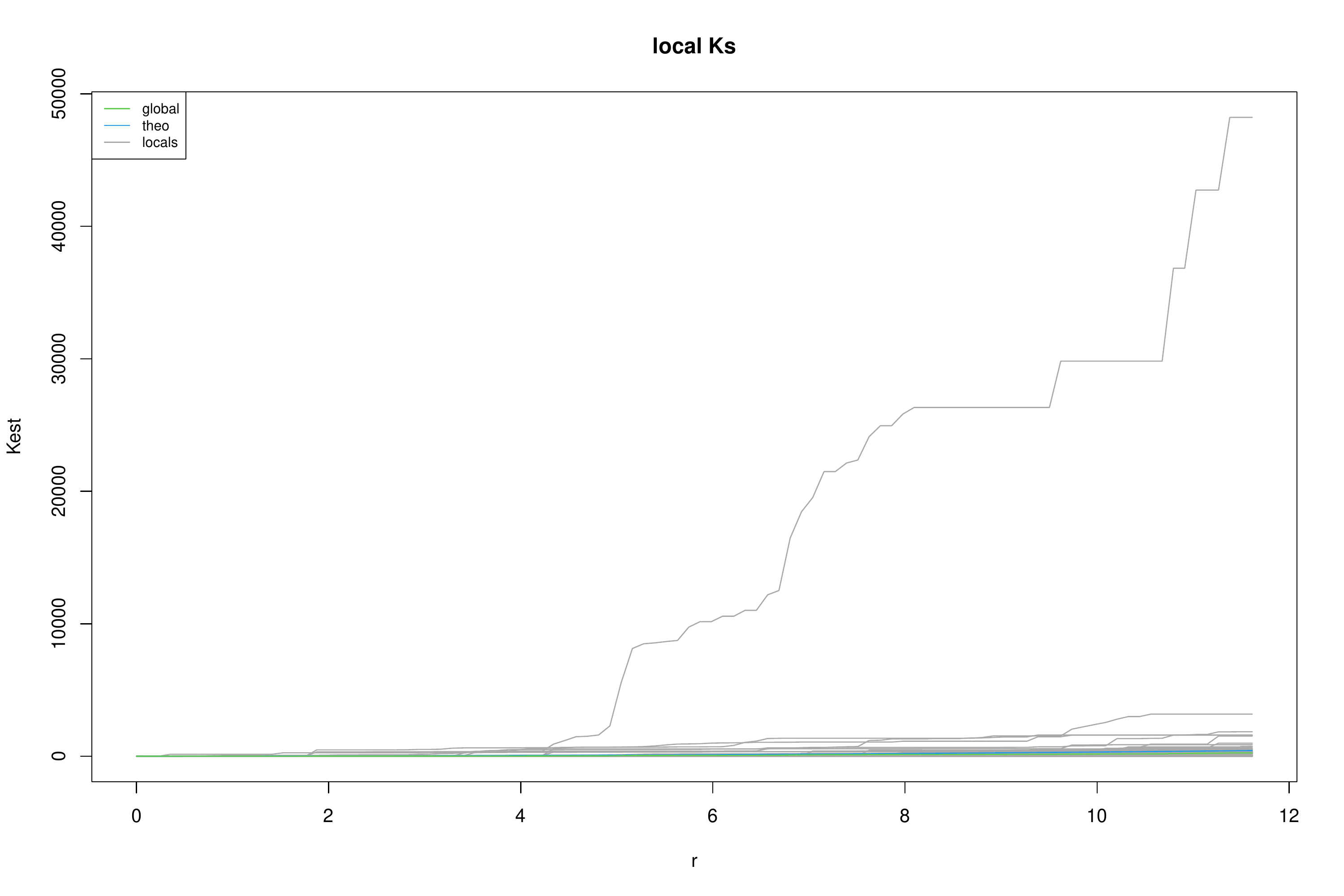}}
		\subfloat{\includegraphics[width=0.33\textwidth]{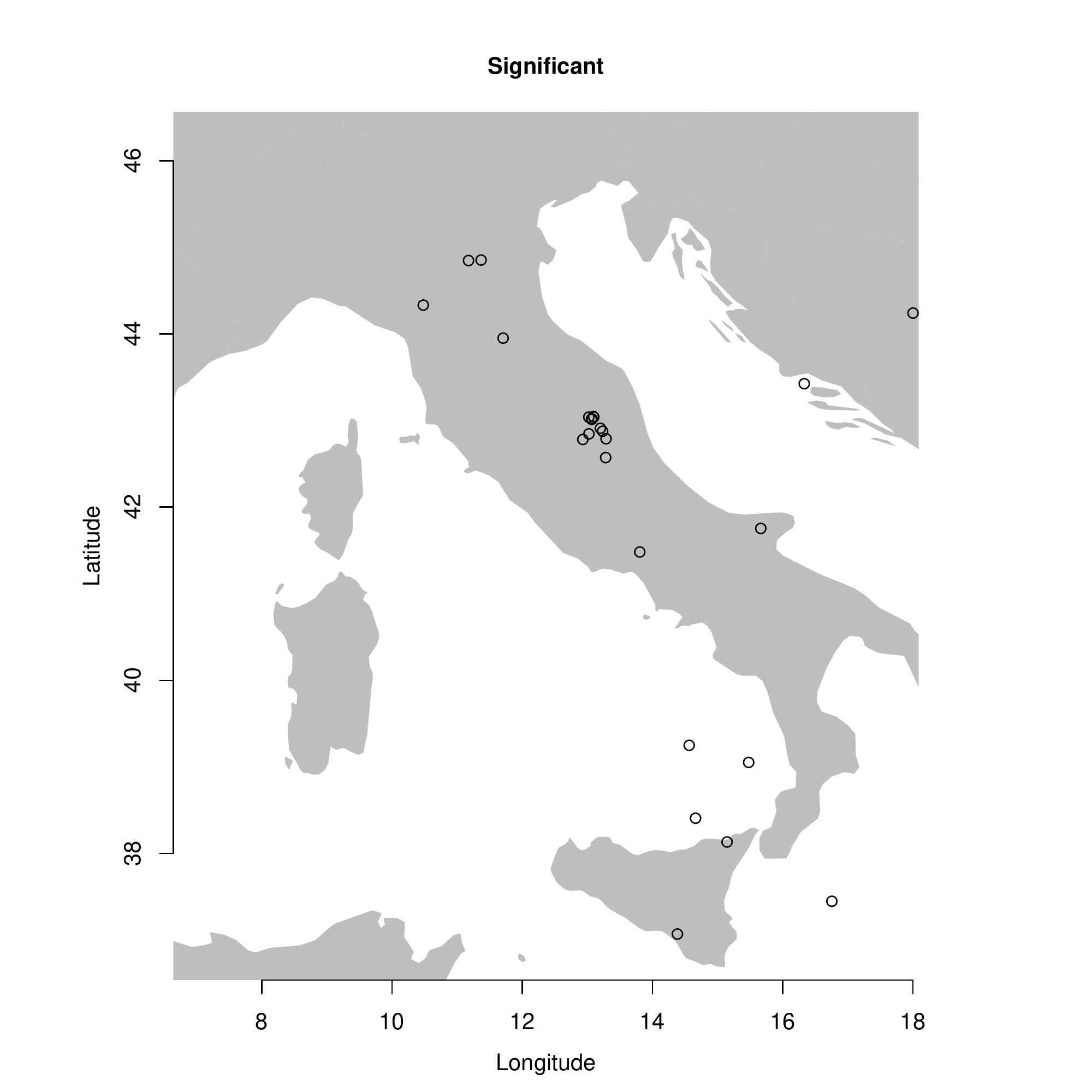}}
	\subfloat{\includegraphics[width=0.33\textwidth]{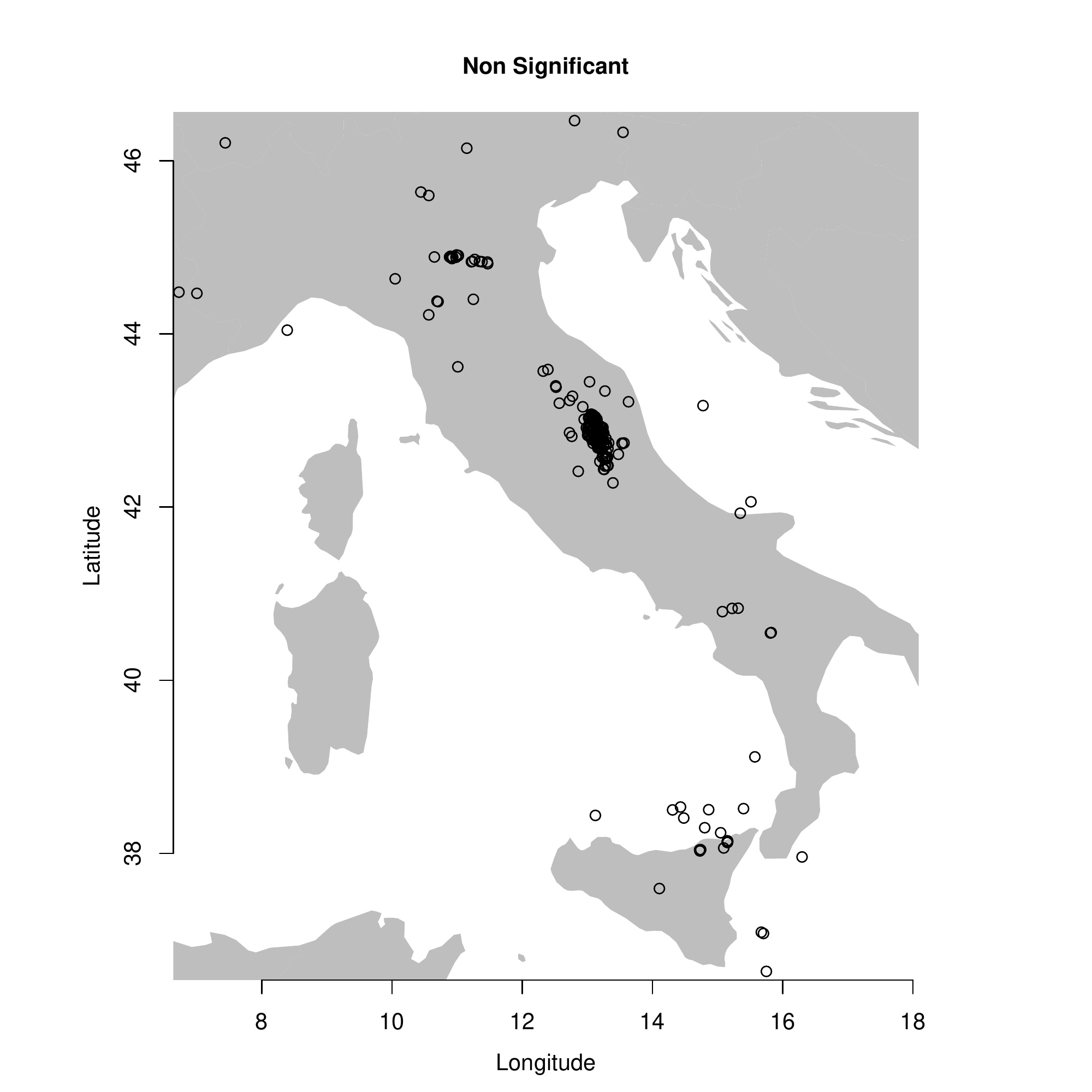}}
	\caption{ \textit{Left panels:} Local $K$-functions. \textit{Center and right panels:}  Results of the local test at $\alpha = 0.1$. \textit{Top panels:} The $K$-function is based on a kernel intensity estimate whose bandwidth is selected by \cite{diggle:13}'s rule. \textit{Bottom panels:} The bandwidth is chosen as in \cite{cronie2018non}.}
	\label{fig:x cubed graph13}
\end{figure}


Figure	\ref{fig:x cubed graph13} displays the significant points (center panels) and the non-significant ones (right panels).
The top panels show the results with \cite{diggle:13}'s bandwidth while the bottom ones are obtained with \cite{cronie2018non}'s bandwidth. For both choices, we selected a 
significance level of $0.1$. 
We observe that the significant points tend to be similar in both cases, therefore the choice of bandwidth (selection method) does not seem to be crucial. We note that such bandwidth-induced differences were missing in the previously run simulation study. We attribute this sensitivity of the procedure to the shapes of the functional marks, that  are obviously more variable, if compared to the simulated ones.

Nevertheless, both bandwidths lead to significant events belonging to important well known Italian seismic sequences. Of course these sequences are likely generated by different underlying processes, giving rise to long-term and highly correlated aftershocks.

\section{Conclusions}
\label{sec:conc}

In this work, we have proposed a general form for local summary statistics for marked point processes, which has been exploited to define the family of local inhomogeneous mark-weighted summary statistics for spatial point processes with functional marks, i.e. Functional Marked Point Processes (FMPP). We have employed such local summary statistics to construct a local test for random labelling, that is, to identify points, as well as regions, where this hypothesis does not hold.

More specifically, we first introduce a general local function for marked point patterns. 
With this specification, we are able to show that this function may be exploited to generate most summary statistics established in the literature.
With particular reference to the functional marked context, we define the family of local $t$-weighted marked $n$-th order inhomogeneous summary statistics based on the $K$-function, which is a local contribution to a global summary statistic estimator. We obtain a result for the expectation of the general local summary statistic and exploit it to derive an expression for the expectation of 
our $t$-weighted local statistics.

Having access to these tools, we have proposed a local test of random labelling, resorting to the second order version of our proposed local estimator, obtaining a local test useful for identifying specific regions where a global test would not detect atypical behaviour of the points.

To study the performance of the test in terms of classification rates, we have conducted a simulation study, considering a number of scenarios with different ground processes and structures for the functional marks. Such simulations have shown that in many settings, the local test performs well in identifying points of a pattern where the hypothesis of random labelling is not verified.

We can draw a number of future work paths.
Nevertheless, the local functions proposed in this paper can be considered as a very informative synthesis of the local second order behavior, useful for characterising the study area by an extended marked model, based on the FMPP theory. Incorporating local characteristics as functional marks would become part of the so called 
\textit{Constructed functional marks} (CFMs), which are marks reflecting the geometries of point configurations in neighbourhoods of the individual points. 

Concerning the application to seismic data, we aim at including also auxiliary (non-functional) marks into the analysis. These could contain synthetic information about the waveforms, such as the arrival times of the seismic event, or the inter-time between the two. 
The achievement of the unification of earthquake data and the FMPP theory would result in building a framework where it would be possible to exploit the available information of the seismic point process altogether.

A final comment concerns the possible extension of this paper's tools to spatio-temporal ground processes, which of course are of importance for processes which typically exhibit spatio-temporal interactions, such as the seismic one.


\begin{thebibliography}{}
	
	\bibitem[Adelfio et~al., 2011]{adelfio2011fpca}
	Adelfio, G., Chiodi, M., D'Alessandro, A., and Luzio, D. (2011).
	\newblock Fpca algorithm for waveform clustering.
	\newblock {\em Journal of Communication and Computer}, 8(6):494--502.
	
	\bibitem[Adelfio et~al., 2012]{adelfio2012simultaneous}
	Adelfio, G., Chiodi, M., D'Alessandro, A., Luzio, D., D'Anna, G., and Mangano,
	G. (2012).
	\newblock Simultaneous seismic wave clustering and registration.
	\newblock {\em Computers \& geosciences}, 44:60--69.
	
	\bibitem[Adelfio et~al., 2020]{adelfio2020some}
	Adelfio, G., Siino, M., Mateu, J., and Rodr{\'\i}guez-Cort{\'e}s, F.~J. (2020).
	\newblock Some properties of local weighted second-order statistics for
	spatio-temporal point processes.
	\newblock {\em Stochastic Environmental Research and Risk Assessment},
	34(1):149--168.
	
	\bibitem[Anselin, 1995]{anselin:95}
	Anselin, L. (1995).
	\newblock Local indicators of spatial association-lisa.
	\newblock {\em Geographical analysis}, 27(2):93--115.
	
	\bibitem[Anselin, 1996]{anselin1996chapter}
	Anselin, L. (1996).
	\newblock Chapter eight the moran scatterplot as an esda tool to assess local
	instability in spatial association.
	\newblock {\em Spatial Analytical}, 4:121.
	
	\bibitem[Baddeley et~al., 2000]{baddeley2000non}
	Baddeley, A.~J., M{\o}ller, J., and Waagepetersen, R. (2000).
	\newblock Non-and semi-parametric estimation of interaction in inhomogeneous
	point patterns.
	\newblock {\em Statistica Neerlandica}, 54(3):329--350.
	
	\bibitem[Chiodi et~al., 2013]{chiodi2013clustering}
	Chiodi, M., Adelfio, G., D’Alessandro, A., and Luzio, D. (2013).
	\newblock Clustering and registration of multidimensional functional data.
	\newblock In {\em Statistical Models for Data Analysis}, pages 89--97.
	Springer.
	
	\bibitem[Chiu et~al., 2013]{chiu:stoyan:kendall:mecke:13}
	Chiu, S.~N., Stoyan, D., Kendall, W.~S., and Mecke, J. (2013).
	\newblock {\em Stochastic Geometry and Its Applications}.
	\newblock John Wiley \& Sons, third edition.
	
	\bibitem[Comas et~al., 2011]{comas2011second}
	Comas, C., Delicado, P., and Mateu, J. (2011).
	\newblock A second order approach to analyse spatial point patterns with
	functional marks.
	\newblock {\em Test}, 20(3):503--523.
	
	\bibitem[Cressie and Collins, 2001]{cressie2001analysis}
	Cressie, N. and Collins, L.~B. (2001).
	\newblock Analysis of spatial point patterns using bundles of product density
	lisa functions.
	\newblock {\em Journal of agricultural, biological, and environmental
		statistics}, 6(1):118--135.
	
	\bibitem[Cronie and Van~Lieshout, 2015]{cronie:lieshout:15}
	Cronie, O. and Van~Lieshout, M. (2015).
	\newblock A {J}-function for inhomogeneous spatio-temporal point processes.
	\newblock {\em Scandinavian Journal of Statistics}, 42(2):562--579.
	
	\bibitem[Cronie and van Lieshout, 2016]{cronie2016summary}
	Cronie, O. and van Lieshout, M. N.~M. (2016).
	\newblock Summary statistics for inhomogeneous marked point processes.
	\newblock {\em Annals of the Institute of Statistical Mathematics},
	68(4):905--928.
	
	\bibitem[Cronie and Van~Lieshout, 2018]{cronie2018non}
	Cronie, O. and Van~Lieshout, M. N.~M. (2018).
	\newblock A non-model-based approach to bandwidth selection for kernel
	estimators of spatial intensity functions.
	\newblock {\em Biometrika}, 105(2):455--462.
	
	\bibitem[Daley and Vere-Jones, 2008]{daley:vere-jones:08}
	Daley, D.~J. and Vere-Jones, D. (2008).
	\newblock {\em An Introduction to the Theory of Point Processes. Volume II:
		General Theory and Structure}.
	\newblock Springer-Verlag, New York, second edition.
	
	\bibitem[D'Angelo et~al., 2021]{dangelo2021assessing}
	D'Angelo, N., Adelfio, G., and Mateu, J. (2021).
	\newblock Assessing local differences between the spatio-temporal second-order
	structure of two point patterns occurring on the same linear network.
	\newblock {\em Spatial Statistics}, 45:100534.
	
	\bibitem[D'Angelo et~al., 2022a]{dangelo2021locally}
	D'Angelo, N., Adelfio, G., and Mateu, J. (2022a).
	\newblock Locally weighted minimum contrast estimation for spatio-temporal
	log-gaussian cox processes.
	\newblock {\em Submitted}.
	
	\bibitem[D'Angelo et~al., 2022b]{dangelo2021locall}
	D'Angelo, N., Siino, M., D'Alessandro, A., and Adelfio, G. (2022b).
	\newblock Local spatial log-gaussian cox processes for seismic data.
	\newblock {\em Advances in Statistical Analysis.
		\url{https://doi.org/10.1007/s10182-022-00444-w}}.
	
	\bibitem[Diggle, 1985]{diggle1985kernel}
	Diggle, P. (1985).
	\newblock A kernel method for smoothing point process data.
	\newblock {\em Journal of the Royal Statistical Society: Series C (Applied
		Statistics)}, 34(2):138--147.
	
	\bibitem[Diggle, 2013]{diggle:13}
	Diggle, P.~J. (2013).
	\newblock {\em Statistical Analysis of Spatial and Spatio-Temporal Point
		Patterns}.
	\newblock CRC Press.
	
	\bibitem[Gabriel and Diggle, 2009]{gabriel:diggle:09}
	Gabriel, E. and Diggle, P.~J. (2009).
	\newblock Second-order analysis of inhomogeneous spatio-temporal point process
	data.
	\newblock {\em Statistica Neerlandica}, 63(1):43--51.
	
	\bibitem[Gelfand et~al., 2010]{gelfand:diggle:guttorp:fuentes:10}
	Gelfand, A.~E., Diggle, P., Guttorp, P., and Fuentes, M. (2010).
	\newblock {\em Handbook of spatial statistics}.
	\newblock CRC press.
	
	\bibitem[Getis and Franklin, 1987]{getis:franklin:87}
	Getis, A. and Franklin, J. (1987).
	\newblock Second-order neighborhood analysis of mapped point patterns.
	\newblock {\em Ecology}, 68(3):473--477.
	
	\bibitem[Getis and Ord, 1992]{getis:ord:92}
	Getis, A. and Ord, J.~K. (1992).
	\newblock The analysis of spatial association by use of distance statistics.
	\newblock {\em Geographical Analysis}, 24(3):189--206.
	
	\bibitem[Ghorbani et~al., 2021]{ghorbani2021functional}
	Ghorbani, M., Cronie, O., Mateu, J., and Yu, J. (2021).
	\newblock Functional marked point processes: a natural structure to unify
	spatio-temporal frameworks and to analyse dependent functional data.
	\newblock {\em Test}, 30(3):529--568.
	
	\bibitem[Iftimi et~al., 2019]{iftimi2019second}
	Iftimi, A., Cronie, O., and Montes, F. (2019).
	\newblock Second-order analysis of marked inhomogeneous spatiotemporal point
	processes: Applications to earthquake data.
	\newblock {\em Scandinavian Journal of Statistics}, 46(3):661--685.
	
	\bibitem[Illian et~al., 2006]{illian2006principal}
	Illian, J., Benson, E., Crawford, J., and Staines, H. (2006).
	\newblock Principal component analysis for spatial point processes—assessing
	the appropriateness of the approach in an ecological context.
	\newblock In {\em Case studies in spatial point process modeling}, pages
	135--150. Springer.
	
	\bibitem[Illian et~al., 2008]{illian:penttinen:stoyan:stoyan:08}
	Illian, J., Penttinen, A., Stoyan, H., and Stoyan, D. (2008).
	\newblock {\em Statistical Analysis and Modelling of Spatial Point Patterns},
	volume~70.
	\newblock John Wiley \& Sons.
	
	\bibitem[Mateu et~al., 2007]{mateu:lorenzo:porcu:07}
	Mateu, J., Lorenzo, G., and Porcu, E. (2007).
	\newblock Detecting features in spatial point processes with clutter via local
	indicators of spatial association.
	\newblock {\em Journal of Computational and Graphical Statistics},
	16(4):968--990.
	
	\bibitem[Mateu et~al., 2010]{mateu:lorenzo:porcu:10}
	Mateu, J., Lorenzo, G., and Porcu, E. (2010).
	\newblock Features detection in spatial point processes via multivariate
	techniques.
	\newblock {\em Environmetrics}, 21(3-4):400--414.
	
	\bibitem[Michelini et~al., 2021]{michelini2021instance}
	Michelini, A., Cianetti, S., Gaviano, S., Giunchi, C., Jozinovi{\'c}, D., and
	Lauciani, V. (2021).
	\newblock Instance--the italian seismic dataset for machine learning.
	\newblock {\em Earth System Science Data}, 13(12):5509--5544.
	
	\bibitem[M{\o}ller, 2003]{moller:03}
	M{\o}ller, J. (2003).
	\newblock Shot noise cox processes.
	\newblock {\em Advances in Applied Probability}, pages 614--640.
	
	\bibitem[M{\o}ller and Ghorbani, 2012]{moller:mohammad:12}
	M{\o}ller, J. and Ghorbani, M. (2012).
	\newblock Aspects of second-order analysis of structured inhomogeneous
	spatio-temporal point processes.
	\newblock {\em Statistica Neerlandica}, 66(4):472--491.
	
	\bibitem[M{\o}ller and Waagepetersen, 2003]{moller:waagepetersen:04}
	M{\o}ller, J. and Waagepetersen, R.~P. (2003).
	\newblock {\em Statistical Inference and Simulation for Spatial Point
		Processes}.
	\newblock Chapman and Hall/CRC, Boca Raton.
	
	\bibitem[Moraga and Montes, 2011]{moraga:montes:11}
	Moraga, P. and Montes, F. (2011).
	\newblock Detection of spatial disease clusters with lisa functions.
	\newblock {\em Statistics in Medicine}, 30(10):1057--1071.
	
	\bibitem[Mrkvi{\v{c}}ka et~al., 2021]{mrkvivcka2021revisiting}
	Mrkvi{\v{c}}ka, T., Dvo{\v{r}}{\'a}k, J., Gonz{\'a}lez, J.~A., and Mateu, J.
	(2021).
	\newblock Revisiting the random shift approach for testing in spatial
	statistics.
	\newblock {\em Spatial Statistics}, 42:100430.
	
	\bibitem[Myllym{\"a}ki et~al., 2017]{myllymaki2017global}
	Myllym{\"a}ki, M., Mrkvi{\v{c}}ka, T., Grabarnik, P., Seijo, H., and Hahn, U.
	(2017).
	\newblock Global envelope tests for spatial processes.
	\newblock {\em Journal of the Royal Statistical Society: Series B (Statistical
		Methodology)}, 79(2):381--404.
	
	\bibitem[Penttinen and Stoyan, 1989]{penttinen1989statistical}
	Penttinen, A. and Stoyan, D. (1989).
	\newblock Statistical analysis for a class of line segment processes.
	\newblock {\em Scandinavian Journal of Statistics}, pages 153--168.
	
	\bibitem[{R Core Team}, 2022]{R}
	{R Core Team} (2022).
	\newblock {\em R: A Language and Environment for Statistical Computing}.
	\newblock R Foundation for Statistical Computing, Vienna, Austria.
	
	\bibitem[Ramsay and Silverman, 2002]{ramsay2002applied}
	Ramsay, J.~O. and Silverman, B.~W. (2002).
	\newblock {\em Applied functional data analysis: methods and case studies},
	volume~77.
	\newblock Springer.
	
	\bibitem[Ripley, 1976]{ripley:76}
	Ripley, B.~D. (1976).
	\newblock The second-order analysis of stationary point processes.
	\newblock {\em Journal of Applied Probability}, 13:255--266.
	
	\bibitem[Schlather, 2001]{schlather2001second}
	Schlather, M. (2001).
	\newblock On the second-order characteristics of marked point processes.
	\newblock {\em Bernoulli}, pages 99--117.
	
	\bibitem[Siino et~al., 2017]{siino2017spatial}
	Siino, M., Adelfio, G., Mateu, J., Chiodi, M., and D’alessandro, A. (2017).
	\newblock Spatial pattern analysis using hybrid models: an application to the
	hellenic seismicity.
	\newblock {\em Stochastic Environmental Research and Risk Assessment},
	31(7):1633--1648.
	
	\bibitem[Siino et~al., 2018]{siino2018testing}
	Siino, M., Rodr{\'\i}guez-Cort{\'e}s, F.~J., Mateu, J., and Adelfio, G. (2018).
	\newblock Testing for local structure in spatiotemporal point pattern data.
	\newblock {\em Environmetrics}, 29(5-6):e2463.
	
	\bibitem[Stoyan and Stoyan, 1994]{stoyan:stoyan:94}
	Stoyan, D. and Stoyan, H. (1994).
	\newblock {\em Fractals, Random Shapes and Point Fields: methods of geometrical
		statistics}.
	\newblock Wiley, Chichester.
	
	\bibitem[Van~Lieshout, 2000]{van2000markov}
	Van~Lieshout, M. (2000).
	\newblock {\em Markov point processes and their applications}.
	\newblock World Scientific.
	
	\bibitem[Van~Lieshout, 2006]{van2006j}
	Van~Lieshout, M. (2006).
	\newblock A j-function for marked point patterns.
	\newblock {\em Annals of the Institute of Statistical Mathematics},
	58(2):235--259.
	
	\bibitem[Van~Rossum and Drake~Jr, 1995]{van1995python}
	Van~Rossum, G. and Drake~Jr, F.~L. (1995).
	\newblock {\em Python reference manual}.
	\newblock Centrum voor Wiskunde en Informatica Amsterdam.
	
\end{thebibliography}
\end{document}